\documentclass[a4paper,12pt]{article}
\usepackage[bbgreekl]{mathbbol}
\usepackage{mathrsfs}
\usepackage{graphicx}
\usepackage{amsmath}
\usepackage{amsfonts}
\usepackage{mathrsfs}
\usepackage{indentfirst}
\usepackage{amsthm}
\usepackage{amssymb}
\usepackage{xcolor}
\usepackage{geometry}
\usepackage{url}
\usepackage{setspace}
\usepackage{verbatim}
\usepackage{ulem}
\usepackage{multirow}
\usepackage{booktabs}
\usepackage[page,toc,titletoc,title]{appendix}
\usepackage{bm}
\usepackage{cases}
\usepackage{subfigure}
\usepackage{subcaption}
\usepackage{amsmath}
\usepackage[noend]{algpseudocode}
\usepackage{algorithmicx,algorithm}
\usepackage{chngcntr}
\usepackage{ragged2e}
\usepackage{slashed}
\usepackage{caption}
\usepackage{upgreek}
\usepackage{yhmath}
\usepackage{bbm}
\usepackage{enumerate}
\usepackage{cancel}
\usepackage{hyperref}

\newtheorem{theorem}{Theorem}[section]
\newtheorem{lemma}{Lemma}[section]
\newtheorem{remark}{Remark}[section]

\renewcommand{\theequation}{\arabic{section}.\arabic{equation}}

\renewcommand{\thelemma}{\arabic{section}.\arabic{lemma}}

\makeatletter
\renewcommand{\@thesubfigure}{\hskip\subfiglabelskip}
\makeatother

\newcommand{\R}{\mathbb{R}}
\newcommand{\C}{\mathbb{C}}
\newcommand{\Cc}{\mathscr{C}}
\newcommand{\Z}{\mathbb{Z}}

\newcommand\dd{~{\rm d}\hskip 0.03cm}

\newcommand\ham{\mathcal{H}}
\newcommand\hH{\hat{\mathcal{H}}}

\newcommand{\im}{{\rm i}}

\newcommand{\dist}{\mathbb{d}}

\newcommand{\resol}{\mathscr{R}}
\newcommand{\Dres}{\delta_\gamma}
\newcommand{\bb}{\pmb{b}}

\newcommand{\Tr}{\mathrm{Tr}}
\newcommand\aTr{\underline{\rm Tr}}

\newcommand{\F}{\mathcal{F}}
\newcommand{\FT}{\mathscr{F}}
\newcommand{\unfold}{\mathcal{T}}
\newcommand{\Sc}{\mathscr{L}}

\newcommand\X{\mathcal{X}}

\newcommand\df{\mathfrak{d}}
\newcommand\setg{\Lambda_{\zeta,\delta}}
\newcommand\bG{{\bf G}}
\newcommand\RL{\mathcal{R}}
\newcommand\bzero{{\bf 0}}
\newcommand\DD{\mathcal{D}}

\newcommand{\lodoss}[3]{\big[g\big(#1\big)\big]_{#2,#3}}

\newcommand{\ucut}{{U_\mathrm{c}}}

\newcommand{\Fk}{K}
\newcommand{\DFkq}{J_{Q,Q'}}
\newcommand{\Fj}{\bar{\jmath}}

\title{Relaxation of Incommensurate Structures via Quantum Models\thanks{
This work was funded by the National Key R \& D Program of China under grant 2025YFA1016600. 
}
}

\author{Mengfan Tu\thanks{{\it tumengfan@mail.bnu.edu.cn}.
School of Mathematical Sciences, Beijing Normal University, Beijing 100875, China
},~~
Huajie Chen\thanks{{\it chen.huajie@bnu.edu.cn}. 
School of Mathematical Sciences, Beijing Normal University, Beijing 100875, China
},~~
Daniel Massatt\thanks{{\it daniel.massatt@njit.edu}.
Department of Mathematical Sciences, New Jersey Institute of Technology, Newark, NJ 07102, USA
}
}\date{}

\begin{document}
\maketitle

\begin{abstract}
Accurately modeling structural relaxation in incommensurate systems is intrinsically challenging due to the absence of global translational symmetry. 
In this work, we develop a variational quantum framework for structural relaxation in incommensurate Schr\"{o}dinger models, where displacement fields are formulated on the configuration space and the electronic Hamiltonian is represented in reciprocal space. 
This yields well-defined relaxed energy, local density of states, and forces through thermodynamic limits.
We propose an anisotropic scattering-channel approximation, and prove exponential convergence of the approximate equilibria. Numerical experiments are performed to support the analysis and show that the model captures domain-wall formation and its impact on the electronic spectrum.
\end{abstract}

\section{Introduction}
\label{sec:intro}

Low-dimensional materials often exhibit emergent physical phenomena due to their strong sensitivity to geometry and external perturbations. 
A prominent example is the discovery of superconductivity in twisted bilayer graphene (TBG) near the magic angle, which has stimulated extensive interest in twisted incommensurate systems \cite{becker2022mathematics, becker2022fine, cao2018correlated, tarnopolsky2019origin, watson2021existence, xu2013graphene}. 
In such systems, lattice mismatch or relative rotation between layers generates long-wavelength moir\'{e} patterns, leading to pronounced spatial modulation of interlayer stacking configurations. 
These moir\'{e} structures play a central role in determining both electronic and mechanical properties \cite{dean2013hofstadter, geim2013vdw}.
To minimize interlayer interaction energy, atomic positions undergo structural relaxation, which further reshapes the moir\'{e} pattern. 
This relaxation effect has been shown to significantly influence the resulting low-energy electronic band structure \cite{carr2018relaxation,dai2016twisted,dean2013hofstadter,fang2019abinitio,geim2013vdw,hott2024allencahn,nam2017lattice,yoo2019reconstruction}. 
It is therefore crucial to develop reliable and accurate models from a theoretical and computational perspective, that can provide a careful treatment of the interplay between geometric relaxation and electronic structure.
However, accurate modeling remains challenging, primarily due to the intrinsic incommensurability, which prevents the direct application of conventional periodic approaches.

The standard approach in electronic structure claculation is the supercell method, which approximates an incommensurate structure by a commensurate periodic system. 
This method becomes computationally prohibitive at small twist angles, where the moir\'{e} period grows rapidly \cite{koda2016coincidence, komsa2013electronic, loh2015graphene,kong2025interact, massatt2023electronic}. 
To address this limitation, continuum models such as the Bistritzer–MacDonald (BM) model provide an effective description at the moir\'{e} scale, capturing key features including flat electronic bands at significantly reduced computational cost \cite{bistritzer2010transport, bistritzer2011moire, watson2023bm}. 
However, these models typically assume rigid lattices and ideal periodicity, and therefore do not account for accurate structural relaxation.
More recently, theoretical frameworks based on configuration space have been developed to treat electronic properties directly on incommensurate geometries, yielding well-defined observables and convergent numerical schemes \cite{massatt2017electronic, wang2025convergence}. 
At the density functional level,  \cite{cances2026kohnsham} established a rigorous Kohn--Sham/reduced Hartree--Fock framework for encapsulated periodic and quasiperiodic two-dimensional materials, showing that electrode-induced screening yields well-posed ground state models for systems including incommensurate moir\'{e} bilayers.
Together, these results provide a rigorous foundation for electronic-structure modeling on genuinely incommensurate geometries.

Nevertheless, when structural relaxation is included, existing approaches commonly decouple mechanical and electronic degrees of freedom. 
These approaches describe the lattice deformation by continuum elasticity with stacking-dependent interlayer energies, leading to reconstruction and domain formation, while electronic properties are subsequently computed on the relaxed structures using tight-binding or related models \cite{fang2019abinitio, kong2025interact, massatt2023electronic}.
In such treatments, the lattice relaxation and electronic modeling are still treated as separate steps, and a consistent framework that determines both relaxed structures and electronic properties within a single description for incommensurate systems remains lacking. 
In particular, there is a notable absence of rigorous mathematical analysis for Schr\"{o}dinger-type operators or density functional methods in this context, leaving the fundamental interplay between variational energy minimization and non-periodic spectral theory largely unexplored.

The goal of this work is to develop a tractable Schrödinger-type relaxation model for incommensurate systems, together with a reciprocal-space algorithm and convergence theory for its numerical solution. 
Motivated by the variational structure of ab initio electronic-structure models, we consider a continuum Hamiltonian whose layer potentials are deformed by the lattice relaxation field. 
The total energy is defined as a spectral functional of this Hamiltonian, in the spirit of Kohn--Sham density functional theory, and the relaxation force is obtained as its variational derivative with respect to the displacement variables. 
This gives a unified framework in which relaxed moir\'e geometries, domain-wall formation, and the associated spectral or band-structure modulation can be treated consistently. 
Moreover, the same variational formulation provides a natural route to phonon-type response calculations by linearizing the force around relaxed configurations.

For the numerical solution, we introduce a scattering-channel discretization adapted to the incommensurate setting and prove its convergence under suitable stability assumptions. 
The numerical experiments validate the convergence theory and show that the model captures the expected moir\'e relaxation behavior and its spectral consequences.

\vskip 0.2cm

\noindent
{\bf Outline.}
The rest of this paper is organized as follows. 
In Section \ref{sec:preliminary}, we formulate the displacement field on the configuration space and derive the reciprocal-space representation of the Hamiltonian on deformed geometries. 
Based on this formulation, we define the relaxed energy, the forces on the deformed geometry, and the variational problem for the equilibrium configuration. 
In Section \ref{sec:planewave}, we propose a  scattering channel method for the variational problem and establish convergence rates with respect to the discretization parameters.
In Section \ref{sec:numerics}, we present numerical experiments to support the analysis and illustrate the physical capabilities of the model. 
Concluding remarks are given in Section \ref{sec:conclude} and detailed proofs are provided in the appendices.

\section{Relaxation of the incommensurate systems}
\label{sec:preliminary}
\setcounter{equation}{0}

This section develops the mathematical framework for the structural relaxation of incommensurate systems. 
We first introduce a compact configuration space to parameterize the infinitely extended displacement fields, establishing a rigorous geometric foundation for lattice deformations. 
Based on this, we define the corresponding Hamiltonian associated with the relaxed displacements and formulate the variational problem to determine the ground state.

\subsection{Reference and displacement}
\label{sec:reference}

Consider two parallel periodic atomic layers in $\R^d$ ($d=1,2$). 
While realistic bilayer systems are separated by a constant out-of-plane distance, this vertical separation acts as a simple tensor product component that does not alter the fundamental analytical properties of the in-plane relaxation. 
For clarity and notational brevity, we restrict our geometric formulation to the in-plane dimension $\R^d$, and restore the out-of-plane degree of freedom in the numerical experiments (see Section \ref{sec:numerics}).

We begin by specifying the rigid reference configuration. Let $A_j \in \R^{d \times d}$ ($j = 1,2$) be non-singular matrices. The geometries of the two periodic layers are characterized by the Bravais lattices $\RL_j = A_j \mathbb{Z}^d$. Let
\begin{equation*}
\RL^\ast_j := 2\pi A_j^{-T}\mathbb{Z}^d, 
\qquad 
\Gamma_j := A_j [0,1)^d, 
\qquad 
\Gamma^\ast_j := 2\pi A_j^{-T}[0,1)^d,
\end{equation*}
where $\RL^\ast_j$ denotes the reciprocal lattice, and $\Gamma_j$, $\Gamma^\ast_j$ represent the corresponding unit cells in real and reciprocal spaces, respectively. 
We denote by $C_{\rm per}(\Gamma_j)$ the space of continuous functions on $\R^d$ that are $\RL_j$-periodic. 

While the individual lattices $\RL_j$ possess discrete translational symmetry, the stacked bilayer system $\RL_1 \cup \RL_2$ generally lacks global periodicity.
We call two lattices $\RL_1$ and $\RL_2$ incommensurate if they share no common translational symmetry, satisfying
\begin{equation*}
\RL_1 \cup \RL_2 + \tau = \RL_1 \cup \RL_2 \quad\Leftrightarrow\quad \tau = \bzero \in \R^d .
\end{equation*}
Conversely, the heterostructure is commensurate if the common translation periods contain $d$ linearly independent vectors. Equivalently, the combined structure admits a periodic supercell.

To account for structural relaxation, we introduce the discrete displacement vectors $U_j : \RL_j \to \R^d$. 
The deformed configuration of the system is given by $y_j : \RL_j \to \R^d~(j=1,2)$ with
\begin{eqnarray}
\label{configuration}
y_j(\ell) = \ell + U_j(\ell) \qquad{\rm for}~\ell\in\RL_j .
\end{eqnarray}
For incommensurate systems, the absence of global translational invariance dictates that the displacement fields $U_j$ cannot be restricted to a finite periodic supercell. 
This intrinsically involves an infinite number of independent degrees of freedom, rendering direct atomistic simulations computationally intractable in real space.

To bypass this infinite dimensionality, we exploit the physical nature of the interlayer interactions, which depend strictly on the local atomic disregistry.
We utilize the ergodicity of the incommensurate interface \cite{cances2017generalized} to map the real-space discrete displacements onto a compact continuous space (on the tori corresponding to the layer unit cells).
For any layer $j \in \{1,2\}$, let $\Fj = 3-j$ denote the index of the opposing layer. 
Each lattice site $\ell \in \RL_j$ is uniquely associated with a relative shift inside the unit cell of the opposing layer $\Gamma_{\Fj}$, through the registry mapping $\bb_j : \RL_j \to \Gamma_{\Fj}$ defined by
\begin{equation*}
\bb_j(\ell) := \ell \bmod \RL_{\Fj} \in \Gamma_{\Fj} .
\end{equation*}
The strict incommensurability ensures that this mapping is injective, effectively uniquely identifying each lattice site $\ell$ with its local structural environment $\bb_j(\ell)$ (see Figure \ref{fig:real_config}). 
This allows us to lift the discrete displacements $U_j(\ell)$ to continuous, periodic functions $u_j : \Gamma_{\Fj} \to \R^d$, such that
\begin{equation}
\label{dispalcement:u}
U_j(\ell) = u_j\big(\bb_j(\ell)\big) \qquad \forall~ \ell \in \RL_j.
\end{equation}
\begin{figure}[htbp!]
\centering
\includegraphics[width=0.95\textwidth]{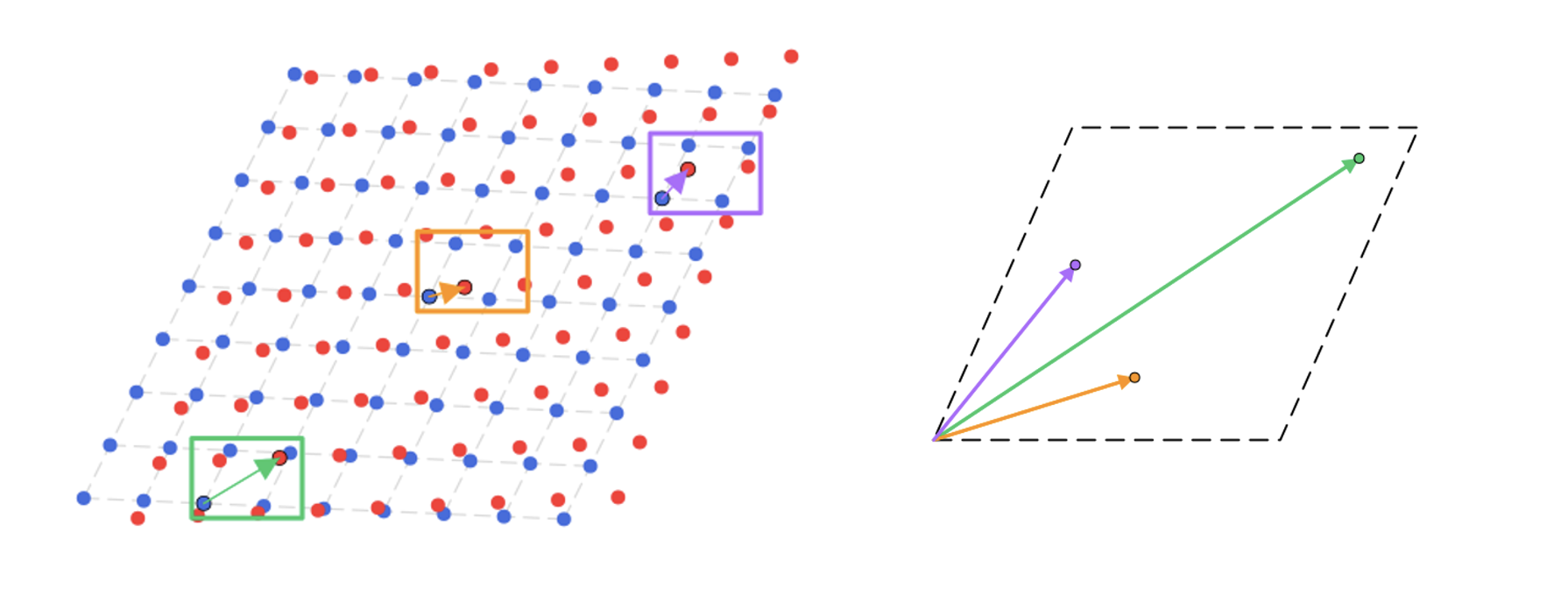}
\caption{Left: A real-space twisted bilayer lattice, where isolated dots and dashed grids represent layers 1 and 2, respectively. 
Right: Illustration of the registry mapping. The vectors $\bb$ in the compact configuration space parameterize the local relative shifts of the first lattice with respect to the second.}
\label{fig:real_config}
\end{figure}

By periodically extending the domain of $u_j$ from $\Gamma_{\Fj}$ to $\R^d$, the displacement fields of the incommensurate system are strictly parameterized within the function space
\begin{align}
\label{space:X}
\X = (\X_1, \X_2) \qquad{\rm with}\quad
\X_j =  \Big\{ u_j \in C_{\rm per}(\Gamma_{\Fj},\R^d) :~ \int_{\Gamma_{\Fj}}  u_j(x) \dd x= 0  \Big\}.
\end{align}
The zero-mean condition in $\X_j$ is introduced to eliminate uniform global translations. 
Since a constant shift of all lattice points preserves the internal geometry and total energy, this restriction ensures a unique internal deformation.
Furthermore, to reflect a physical system dominated by local interlayer coupling rather than long-range interactions, we define a subspace $\mathcal{X}^\gamma \subset \mathcal{X}$ consisting of smooth displacement fields, whose Fourier coefficients exhibit typical exponential decay
\begin{equation}
\label{set:X_gamma}
\mathcal{X}^\gamma := \Big\{ (u_1, u_2) \in \mathcal{X} : \; u_j = \sum_{Q \in \RL^\ast_{\Fj}\setminus\{0\} } \hat{u}_Q e^{\im Q \cdot x}, \text{ and } |\hat{u}_Q| \le C e^{-\gamma |Q|} \text{ for } j=1,2 \Big\}.
\end{equation}
The exclusion of the zero mode ($Q \neq 0$) in the Fourier expansion is the direct equivalent of the zero-mean condition, explicitly ensuring that uniform global translations are removed from the representation.
This configuration space formulation provides an exact and computable descriptor for the incommensurate geometry.

\subsection{Hamiltonian of the deformed systems}
\label{sec:hamiltonian}

We will establish a Schr\"{o}dinger-type model for the incommensurate system with the displacement field $u = (u_1, u_2) \in \mathcal{X}^\gamma$. In this work we focus on the simplified case where the electronic Hamiltonian is linear for a given mechanical relaxation pattern as opposed to the non-linearity in electron density in the DFT framework. 
Specifically, for each layer $j \in \{1, 2\}$, we define the effective potential $V_j[u_j] : \R^d \to \R$ as a superposition of localized, spherically symmetric site potentials
\begin{equation}
\label{V:potential}
V_j[u_j](x) = \sum_{\ell \in \RL_j} v_j\big(x-\ell-u_j(\ell)\big),
\end{equation}
where $v_j:\R^d\rightarrow\R~(j=1,2)$ are sufficiently smooth potentials with exponential decay both in real space and the Fourier domain
\begin{equation}
\label{decay:v}
|v_j(x)| \le C e^{-a |x|},~|\hat v_j(\xi)| \le C e^{-\alpha |\xi|}
\quad{\rm with}\quad
\hat v_j(\xi) := \int_{\R^{d}} v_j(x) e^{-\im \xi \cdot x} \dd x 
\end{equation}
and positive constants $C, \alpha$. 
The associated Schr\"{o}dinger-type operator for the incommensurate bilayer system is then given by
\begin{equation}
\label{H}
\ham[u] := -\frac{1}{2}\Delta + V_1[u_1] + V_2[u_2].
\end{equation}

In the absence of relaxation (i.e. $u = 0$), the $\RL_j$-periodicity of the potentials $V_j$ enables the treatment of the incommensurate structure via a higher-dimensional projection framework \cite{wang2025convergence}. 
Conversely, for the relaxed configurations where $u \in \mathcal{X}^\gamma$, the potential $V_j$ lacks lattice periodicity, as the displacement field is modulated over the configuration space $\Gamma_{\Fj}$. This discrepancy between the lattice symmetry and the displacement field renders the system intractable within traditional real-space or plane-wave computational frameworks.

To rigorously restore a periodic analysis framework, we lift the relaxed problem into an extended superspace $\R^d \times \R^d$. 
By exploiting the $\Gamma_{\Fj}$-periodicity of $u_j$, we construct an auxiliary potential $W_j[u_j] : \R^d \times \R^d \to \R$ defined as
\begin{equation}
\label{W:auxiliary}
W_j[u_j](x_j,x_{\Fj}) := \sum_{\ell \in \RL_j} v_j\big(x_j-\ell-u_j(x_{\Fj}-x_j+\ell)\big) .
\end{equation}
This construction captures the quasi-periodicity of the relaxed system. 
The following lemma ensures the existence of the auxiliary potential $W_j$ as a continuous periodic lift. Its proof is deferred to \ref{sec:proof:lemma:potential}.

\begin{lemma} 
[Properties of the auxiliary potential]
\label{lemma:potential}
For $j=1,2$ and $u \in \X$, the auxiliary potentials $W_j$ satisfy the following properties:
\begin{enumerate}
    \item[(1)] 
    $W_j[u_j](x_j,x_j) = V_j[u_j](x_j)$;
    \item[(2)] 
    $W_j[u_j](x_j,x_{\Fj})$ is periodic with respect to $\Gamma_j \times \Gamma_{\Fj}$;
    \item[(3)] 
    $W_j[u_j](x_j,x_{\Fj})$ is continuous on $\R^d\times\R^d$.
\end{enumerate}
\end{lemma}

By virtue of the joint periodicity and continuity established above, the auxiliary potentials $W_j$ admit a well-defined Fourier series expansion. 
To simplify the notation in the extended superspace, we uniformly assign the joint coordinates $\mathbf{x} = (x_1, x_2) \in \R^d \times \R^d$. 
The joint interaction space is given by $\Gamma_1 \times \Gamma_2$, with its associated reciprocal lattice being $\RL_1^\ast \times \RL_2^\ast$. 
For $j = 1, 2$, the expansion takes the uniform form
\begin{align}
\label{hat:W}
W_j[u_j](x_1, x_2) & = \sum_{G_1 \in \RL^*_1} \sum_{G_2 \in \RL^*_2} \hat{W}_j[u_j](G_1, G_2) e^{\im (G_1 \cdot x_1 + G_2 \cdot x_2)} \qquad \text{with} \nonumber 
\\[1ex]
\hat{W}_j[u_j](G_1, G_2) & = \frac{1}{|\Gamma_1||\Gamma_2|} \int_{\Gamma_1} \int_{\Gamma_2} W_j[u_j](x_1, x_2) e^{-\im (G_1 \cdot x_1 + G_2 \cdot x_2)} \dd x_2 \dd x_1. 
\end{align}

With the Fourier components of the auxiliary potentials established, we now transform the real-space Hamiltonian $\ham[u]$ into reciprocal space. 
Since the incommensurate scattering essentially couples plane-wave states with momenta differing by vectors $G_1 + G_2$ for $(G_1, G_2) \in \RL_1^* \times \RL_2^*$, it is natural to introduce a generalized Bloch wavevector $\xi \in \R^d$. 
We then define the shifted reciprocal-space Hamiltonian $\hH[u](\xi)$, acting on the sequence space $\C(\RL_1^* \times \RL_2^*)$, with its matrix elements given by
\begin{equation} 
\label{Hhat:xi}
\hH[u](\xi)_{\bG,\bG'} = \frac{1}{2}|G_1+G_2+\xi|^2\delta_{G_1G_1'}\delta_{G_2G_2'}+\sum_{j=1,2}\hat{W}_j[u_j](G_1-G_1',G_2-G_2') ,
\end{equation}
where $\bG=(G_1,G_2)$ and $\bG'=(G_1',G_2') \in \RL^*_1\times\RL^*_2$. 
Here, the continuous parameter $\xi$ serves as a generalized quasimomentum. In the incommensurate limit, the set of all possible combinations $G_1 + G_2$ densely fills $\R^d$, and $\xi$ effectively parameterizes the continuous spectrum of the underlying quasi-periodic system.

\begin{remark}
[Intertwining relation between real and reciprocal spaces] 
To rigorously elucidate the connection between $\ham[u]$ in \eqref{H} and $\hH[u](\xi)$ in \eqref{Hhat:xi}, we define the unfolding operator (or, generalized Bloch transform) $\mathcal{T}_\xi$ for $\xi \in \R^d$ on the Schwartz space. 
For $\psi \in \Sc(\R^d)$, the operator $\mathcal{T}_\xi : \Sc(\R^d) \to \C(\RL_1^* \times \RL_2^*)$ is defined by evaluating the Fourier transform $\hat{\psi} = \int_{\R^d} e^{-i\xi\cdot x}\psi(x) \dd x$ on the shifted reciprocal lattice
\begin{equation*}
(\mathcal{T}_\xi \psi)(\bG) := \hat{\psi}(\xi + G_1 + G_2) \qquad \forall~ \bG = (G_1, G_2) \in \RL_1^* \times \RL_2^* .
\end{equation*}
A direct calculation (see \ref{sec:proof:connection}) demonstrates that $\mathcal{T}_\xi$ elegantly intertwines the actions of the real-space and reciprocal-space Hamiltonians via the relation
\begin{equation} 
\label{connection}
\mathcal{T}_\xi \big( \ham[u]\psi \big) = \hH[u](\xi) \big( \mathcal{T}_\xi \psi \big) \qquad \forall ~\psi \in \Sc(\R^d). 
\end{equation}
\end{remark}

\subsection{The variational problem}
\label{sec:energy:force}

The physical observables needed to obtain mechanical relaxation of the incommensurate system are determined by the density of states (DoS) of the Hamiltonian operator $\ham[u]$. 
More precisely, macroscopic observables are given by the ``trace per unit volume" of $g(\ham[u])$, where $g$ belonging to the Schwartz space is related to the observables under consideration.
In this work, to ensure the rigorous convergence of the operator trace in the thermodynamic limit, we consider test functions possessing stronger regularity. 
For positive constants $\zeta > 0$ and $\delta > 0$, we consider the space $\setg$ consisting of functions that admit an analytic continuation into a strip around the real axis with exponential decay
\begin{multline}
\label{deftestfunctiong}
\qquad
\setg := \Big\{ g \in \Sc(\R) : \; g \text{ admits an analytic continuation to } S_\delta,
\\
\text{ and } |g(z)| \le C e^{-\zeta |\mathrm{Re} z|} \;\; \forall~z \in S_\delta \Big\}
\qquad\qquad
\end{multline}
with $S_\delta := \{z \in \mathbb{C}, |\mathrm{Im} z| \le \delta\}$ and  $\Sc(\R)$ the Schwartz space, the set of all rapidly decreasing smooth functions.

A prototypical observable is the system energy, inherently related to the Fermi-Dirac distribution $f_{FD}(\lambda) = \lambda \big(1 + e^{\beta(\lambda-e_{\rm F})}\big)^{-1}$ for a given chemical potential $e_{\rm F} \in \R$ and inverse temperature $\beta > 0$. 
Although this function does not naturally decay as $\lambda \to -\infty$, the Hamiltonian $\ham[u]$ is bounded from below. 
This lower-boundedness allows us to smoothly truncate $f_{FD}$ below the spectral lower bound, such that the modified function satisfies the requisite exponential decay without altering the evaluation of the operator trace.

To define the energy of the incommensurate system, we introduce a functional based on the ``averaged trace'' of $g(H)$, in accordance with the standard framework for extended systems over $\R^d$. 
Formally, this functional is constructed by restricting the operator to a bounded domain $B_R$ and taking the asymptotic limit of the trace per unit volume as $R\rightarrow\infty$. 
To facilitate both the analytical and computational treatments of this limit, we employ a smooth partition of unity $\{\chi_j\}_{j\in \Z^d}$ on $\R^d$. Specifically, let $\chi \in C^{\infty}(\R^d)$ be a reference cut-off function satisfying
\begin{equation*}
\label{partitionofunity}
\chi \geq 0, \qquad
\chi(x) = \left\{  
\begin{array}{lr}  
1,  \quad \|x\|_\infty \leq\frac{1}{3}
\\[1ex]  
0,  \quad \|x\|_\infty >1
\end{array}
\right. 
\qquad {\rm and} \qquad
\sum_{j\in \Z^d}\chi(x-j) \equiv 1 .
\end{equation*} 
With the shifted profiles defined by $\chi_j(x):=\chi(x-j)$, the regularized averaged trace on $B_R$ for any $g\in\setg$ and a given $R>0$ is expressed via the following localized sum:
\begin{equation}
\label{def:dosR}
\aTr_R\Big(g\big(\ham[u]\big)\Big) := \frac{1}{|B_R|}\sum_{j,k\in \Z^d\cap B_R}{\rm Tr}\Big(\chi_j g\big(\ham[u]\big) \chi_k\Big) . 
\end{equation}

The following theorem establishes the convergence of the thermodynamic limit for the averaged trace \eqref{def:dosR} and provides an explicit representation in terms of the local density of states in reciprocal space. 
The proof is deferred to \ref{sec:proof:theorem:Dos}.

\begin{theorem}
[Thermodynamic limit of the DoS]
\label{theo:TDL}
Let $u\in\X^{\gamma}$ and $g \in \setg$, then the thermodynamic limit 
\begin{align}
\label{energy:R}
E[u]:=\lim_{R\rightarrow\infty} \aTr_R\big(g(\ham[u])\big)
\end{align}
exists, and the ground state energy satisfies
\begin{align}
\label{energy}
E[u] = \int_{\R^d} \bigg[g\big(\hH[u](\xi)\big)\bigg]_{\bzero,\bzero} \,\mathrm{d}\xi ,
\end{align}
where $[\cdot]_{\bzero, \bzero}$ denotes the diagonal matrix element at the reciprocal lattice vector $\bG = \bzero$.
\end{theorem}

Note that the integrand $\big[g(\hH[u](\xi))\big]_{\bzero,\bzero}$ in \eqref{energy} gives the local density of states (LDoS) in reciprocal space.
Based on the energy functional $E[u]$ characterized in Theorem~\ref{theo:TDL}, the determination of structural relaxation relies on computing its variation with respect to $u$. 
However, a direct differentiation of $E[u]$ yields highly complex expressions that are computationally intractable. While this gradient admits a formal representation via resolvents and contour integration, the formulation remains heavily implicit. 
The next theorem resolves this complexity by establishing a ``Hellmann-Feynman" like formula that avoids the need for explicit eigenvector derivatives, yielding a highly tractable form for the forces. The proof is provided in \ref{sec:proof:force}.

\begin{theorem}
[Tractable force representation]
\label{theorem:force}
Let $u\in\X^{\gamma}$, $g \in \setg$ and the energy $E[u]$ be given by \eqref{energy}. 
Let $\nabla_u := \{ \partial_{\hat u_{Q}} \}_{k\in (\RL^*_{1} \cup \RL^*_{2})\setminus \{0\} }$ denote the derivatives with respect to the displacement field in the configuration space,
then
\begin{align}
\label{force}
\nabla_{u} E[u] = \int_{\R^d} \bigg[ g'\Big(\hH[u](\xi)\Big) ~ \Big(\nabla_{u} \hH[u](\xi)\Big) \bigg]_{\bzero,\bzero} \dd \xi.
\end{align}
\end{theorem}

The representation \eqref{force} offers substantial computational advantages by expressing the force directly through the derivative $g'$, thereby eliminating the need for explicit eigenvector derivatives or complex contour integrals. 
As a result, evaluating the force incurs a computational cost comparable to that of the energy calculation itself. This identity thus constitutes an analogue of the Hellmann–Feynman theorem tailored for incommensurate systems.

The determination of the physical ground state is then formalized as a variational problem, where the equilibrium configuration of the bilayer system is obtained by solving
\begin{equation}
\label{var:0}
\min_{u\in \X} E[u] .
\end{equation}
For convergence analysis of the numerical approximations, we assume the following regularity and stability conditions on the equilibrium states. 
In the present setting, we are not able to rigorously justify this assumption because the electronic contribution has complicate non-local coupling to the relaxation field.

\begin{flushleft}
{\bf (A)} 
The variational problem \eqref{var:0} has a minimizer $\bar{u}$ such that 
\begin{equation}
\label{regularity_stability}
\bar{u} \in \mathcal{X}^\gamma
\qquad{\rm and}\qquad
\langle \nabla^2 E[\bar{u}] v, v \rangle \ge \vartheta |v|^2 \quad \forall~v \in \X
\end{equation}
with some constants $\gamma>0$ and $\vartheta>0$.
\end{flushleft}

There are two parts in the assumption \eqref{regularity_stability}, which play different roles. 
The condition $\bar u\in\mathcal X^\gamma$ is a regularity assumption on the relaxed displacement field. 
It means that the Fourier coefficients of $\bar u$ decay at a rate determined by $\gamma$, and hence that the equilibrium deformation is smooth across the configuration space.
This is expected when the relaxation is governed mainly by local registry information, and hence distant atoms with similar local configuration exhibit nearly identical local relaxations.

The second condition in \eqref{regularity_stability} is a stong stability condition, which prevents small perturbations of the displacement field from changing the energy only at higher order. 
This condition is not merely technical. 
It is the mechanism that converts consistency estimates for the truncated energy into convergence of the corresponding minimizers. 
Without such a lower bound, the approximate energy could approximate $E$ accurately while its minimizers drift along nearly flat directions. 
Coercivity assumptions of this type are standard in the analysis of stable equilibria, continuum limits, and Cauchy--Born type approximations \cite{carr2018relaxation,chen2019geometry,massatt2023electronic,ortner2013cauchyborn}.

\begin{remark}[Regularity and penalization]
\label{rem:penalized-regularity}
The condition \(\bar u\in \mathcal X^\gamma\) in {\bf (A)} excludes relaxation patterns with substantial high-frequency content.
This restriction may fail in physically relevant regimes, for example, strong interlayer coupling, out-of-plane relaxation, twisted dislocation structures, or proximity to commensurate--incommensurate transitions may produce narrow transition layers or sharp domain walls \cite{carr2018relaxation,massatt2023electronic}. 
These mechanisms generate significant high-frequency components in \(\bar u\), and the Fourier coefficients of \(\bar u\) need not exhibit the exponential decay required by \(\mathcal X^\gamma\). 
In limiting regimes with increasingly sharp transition regions, the decay may be only subexponential, or even algebraic.

A natural way to work within the regularity class assumed here is to introduce a small reciprocal-space penalization. In the product reciprocal lattice
\(\mathcal R_1^\ast \times \mathcal R_2^\ast\),
the combination \(G_1+G_2\) is associated with the physical scattering channel, whereas large values of \(G_1-G_2\) correspond to rapid oscillations of the relaxation field across configuration space. Thus, schematically, one may add to the reciprocal Hamiltonian a regularizing term of the form
\[
\varepsilon |G_1-G_2|^2 \qquad{\rm with}\quad \varepsilon>0 .
\]
This suppresses highly oscillatory displacement modes and is expected to produce a minimizer \(\bar u_\varepsilon\) with improved Fourier decay.
Therefore, the assumption {\bf (A)} can be viewed either as a regularity hypothesis on the physical relaxation field or as the natural regularity class of a mildly penalized model \cite{Li2026Spectral,quan2026scf}. 
A careful study of the peneralized model will be provided in our other work. 
\end{remark}

\section{Scattering-channel approximations}
\label{sec:planewave}
\setcounter{equation}{0}

The minimization problem  \eqref{var:0} is infinite dimensional in two different aspects. 
First, the displacement field $u\in\X$ is a function on the configuration space and therefore contains infinitely many Fourier modes.
Second, for each admissible displacement $u$ and each $\xi \in \R^d$, the reciprocal-space Hamiltonian $\hat \ham [u](\xi)$ is an infinite matrix acting on $\ell^2 \big(\RL_1^* \times \RL_2^*\big)$. 
This section introduces two corresponding numerical approximations: a Fourier truncation of the displacement field and a Scattering-channel truncation of the reciprocal Hamiltonian. 
The resulting approximation is a finite-dimensional variational problem while preserving the reciprocal-space structure of the incommensurate system.

We first discretize the displacement field. Since the displacement is periodic with respect to the configuration variables of each layer, it is natural to expand it in Fourier modes on the corresponding reciprocal lattices. 
For a cutoff $K>0$, we define the finite-dimensional subspace
\begin{align}
\label{var:X:ucut}
\X_{\ucut} := \Big\{ u^\ucut = (u_{1},u_{2})\in\X :~
u_j(x) = \sum_{Q \in \RL_j^*,~ \vert Q \vert \leq \ucut} \hat{u}_{Q} e^{i Q \cdot x} ~{\rm for}~j=1,2 \Big\} .
\end{align}
The parameter $K$ fixes the number of variational degrees of freedom and determines the highest spatial frequency retained in the relaxation field. 

We next truncate the reciprocal-space Hamiltonian. For a fixed displacement $u\in\X$ and $\xi\in\R^d$, the reciprocal-space Hamiltonian $\hH[u](\xi)$ is an infinite matrix indexed by wave-vector pairs in
$\RL_1^*\times\RL_2^*$. 
For $\bG=(G_1,G_2)\in\RL_1^*\times\RL_2^*$, the combination $G_1+G_2$ represents the physical momentum entering the kinetic energy, whereas $G_1-G_2$ measures the relative direction in the product reciprocal lattice. 
Hence, at fixed physical momentum, large $|G_1-G_2|$ corresponds to highly oscillatory scattering channels in configuration space.
Because these two directions are governed by distinct convergence mechanisms, we use the anisotropic cutoff domain as in \cite{wang2025convergence}
\begin{align}
\label{D_WL}
\DD_{W,L}:= \Big\{ \big(G_1,G_2\big) \in\RL_1^* \times \RL_2^* :~ 
\big|G_1+G_2\big|\leq W, ~ \big|G_1-G_2\big|\leq L \Big\} .
\end{align}
Here, the cutoff $W$ suppresses high physical frequencies and the cutoff $L$ controls the incommensurate sampling. 
While single isotropic cutoff for both directions would either waste many degrees of freedom or fail to resolve the more restrictive direction, the asymmetric truncation $\mathcal D_{W,L}$ allows the two sources of error to be balanced separately.
A schematic illustration of this domain in the one-dimensional case is shown in Figure \ref{fig:H_cutoff}.

\begin{figure}[htbp!]
\centering
\includegraphics[width=0.5\textwidth]{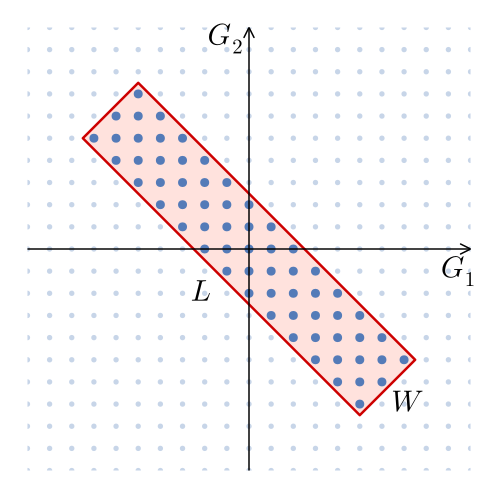}
\caption{The domain $\DD_{W,L}$ in the reciprocal space.}
\label{fig:H_cutoff}
\end{figure}

Restricting $\hH[u](\xi)$ to this finite index set gives the finite matrix
\begin{align}
\label{H:discrete}
\hH[u](\xi)^{\DD_{W,L}} = \big( \hH[u](\xi)_{\bG,\bG'} \big)_{\{\bG,\bG'\} \in \DD_{W,L}}.
\end{align}
Accordingly, the corresponding discrete energy functional is formulated as
\begin{align}
\label{energy:discrete}
E^{W,L}[u] = \int_{\R^{d}} \Big[ g\big(\hat \ham^{\DD_{W,L}} [u](\xi)\big) \Big]_{\bzero,\bzero} \dd \xi.
\end{align}
In this formulation the integral over $\xi$ is kept exact, and a further quadrature discretization can be introduced separately if needed.

Combining the displacement discretization \eqref{var:X:ucut} with the Hamiltonian truncation \eqref{H:discrete}, the continuous variational problem \eqref{var:0} is approximated by the finite-dimensional problem
\begin{equation}
\label{variation:WL_U}
\min_{u\in\X_K} E^{W,L}[u].
\end{equation}

In the numerical implementation, we solve \eqref{variation:WL_U} using the limited-memory Broyden--Fletcher--Goldfarb--Shanno (L-BFGS) quasi-Newton method \cite{nocedal2006numerical} with a backtracking line search.
The required gradient of the energy functional $\nabla_{u}^{W,L}E[u]$ is obtained as the discrete analogue of \eqref{force}. 

The following theorem states that the minimizers of the fully discrete problem \eqref{variation:WL_U} converge to the minimizer of the continuous problem \eqref{var:0} and does not include the optimization error of the iterative solver.
The estimate separates the three approximation errors: the physical scattering-channel cutoff $W$, the incommensurate-direction cutoff $L$, and the displacement cutoff $K$.
The proof of this theorem is presented in \ref{sec:proof:convergence}.

\begin{theorem}
[Convergence of the scattering-channel approximation]
\label{theorem:u:convergence}
Let $g \in \setg$, $\bar{u}\in\X^{\gamma}$ be the solution of \eqref{var:0}. 
If Assumption {\bf (A)} is satisfied, then for parameters $W, L,$ and $\ucut$ sufficiently large, there exists a solution $\bar{u}^{W,L,\ucut}\in\X_{\ucut}$ to \eqref{variation:WL_U}.
Furthermore,  there exist positive constants $C$ and $c$ independent of $W, L, \ucut$, such that 
\begin{align}
\label{u:convergence}
\big\| \bar{u} - \bar{u}^{W,L,\ucut} \big\| \leq \frac{C}{\vartheta} \Big( e^{-c \zeta W} + e^{-c \min{\{\tilde\gamma,\delta\}} L} +e^{-c \min{\{\tilde\gamma,\delta\}} \ucut} \Big) ,
\end{align}
where $\tilde\gamma>0$ is a constant depending on $\gamma$ (as established in Lemma \ref{lemma:decay of W}). 
\end{theorem}

\begin{remark}
[Convergence dependency]
\label{rem:testg_regularity}
The exponential error estimate in \eqref{u:convergence} is governed by two distinct mathematical mechanisms: the analytic properties of the observable function $g \in \setg$ and the spatial regularity of the structural relaxation $\bar{u}\in\X^{\gamma}$. 
First, the convergence with respect to the cutoff $W$ is controlled by $\zeta$, which characterizes the decay of high physical frequencies. 
Second, the decay rates with respect to the cutoff $L$ and the displacement field cutoff $\ucut$ are controlled by the exponent $\min\{\tilde\gamma, \delta\}$, where $\tilde\gamma$ and $\delta$ depend on the regularities of the equilibrium deformation $\bar{u}$ and the observable function $g$ respectively.
Regardless of these parameter dependencies, the overall estimate \eqref{u:convergence} guarantees an exponential decay of the numerical error with respect to all three truncation cutoffs $W$, $L$, and $\ucut$.
\end{remark}

\section{Numerical experiments}
\label{sec:numerics}
\setcounter{equation}{0}

In this section, we perform numerical experiments for incommensurate relaxation problems.
All simulations are performed on a workstation with an 8-core Intel Core i9 processor and 32 GB RAM, using the open-source {\tt Julia} packages {\tt BilayerDFT.jl} \cite{git:Incommensurate_relax}. 

We consider a bilayer system formed by two periodic atomic chains aligned in the $x$-direction.
The two chains have lattice constants $L_1$ and $L_2$ respectively, and are separated by a vertical distance $\dist$ in the $z$-direction. The two layers are placed at $z_1=-\frac{\dist}{2}$ and $z_2=\frac{\dist}{2}$.
The Hamiltonian is defined on the $(x,z)$-plane and is given by 
\begin{align*}
H = -\frac{1}{2}\Big(\frac{\partial^2}{\partial x^2} + \frac{\partial^2}{\partial z^2}\Big) + \sum_{j\in\{1,2\}} \sum_{\ell \in \RL_j} v_j\big(x-\ell,z-z_j\big)
\end{align*}
where the potential $v_j~(j=1,2)$ are given by the Gaussian functions
\begin{align}
\label{siteP:numerics}
v_j(x,z) = -\frac{1}{2\pi \sigma_{j,x} \sigma_{j,z}} \exp  {\left( -\frac{x^2}{2\sigma_{j,x}^2} -\frac{z^2}{2\sigma_{j,z}^2} \right)} 
\end{align}
with corresponding width parameters $\sigma_{1,x}=\sigma_{1,z}=0.27$ and $\sigma_{2,x}=\sigma_{2,z}=0.3$.
To compute the energy of the system, we take $g$ from a Fermi-Dirac distribution
\begin{align*}
g_{\beta,\mu}(\lambda) = \lambda f_{\rm FD}(\lambda) = \frac{\lambda}{1 + e^{\beta(\lambda - \mu)}},
\end{align*}
where the chemical potential is fixed at $\mu = 5.0$ and $\beta$ denotes the inverse temperature.
The parameter $\beta$ controls the sharpness of the observation function, larger $\beta$ corresponds to lower temperature and a less regular spectral function, while smaller $\beta$ gives a smoother occupation profile.

\vskip 0.2cm
{\bf Regularity in configuration space.}
We first check the regularity of the equilibrium displacement field in the configuration space.
We present the relaxed displacement fields together with the decay of their Fourier coefficients in \ref{fig:u_compare_dz}. 
The relaxation for three different interlayer distances $\dist=2.8,~3.2,~3.6$ are compared in this figure, where we set the lattices constants $L_1=1.5$, $L_2=\sqrt{2}*1.04 \approx 1.47$, the inverse temperature $\beta=1.0$, and the scattering-channel cutoffs $L=40$, $W=10$, and $\ucut=200$. 
The top two rows show the two components of the relaxed displacement field, while the bottom two rows show the magnitudes of their Fourier coefficients.
The plots show that the Fourier coefficients decay exponentially in all three cases, which indicates that the equilibrium displacement is dominated by low-frequency modes. 
This provides numerical evidence that the relaxation field is smooth in the configuration space and is therefore well approximated by a finite Fourier expansion. 

Moreover, we observe that the decay becomes faster as the interlayer distance increases. 
This behavior is consistent with the physical interpretation of the model: smaller layer separation produces stronger interlayer relaxation effects and can introduce sharper variations in the displacement field, whereas larger separation weakens the relaxation interaction and leads to smoother equilibrium configurations. 
These observations support the regularity assumption used in the convergence analysis, at least in the parameter regime tested here.

\begin{figure}[htbp!]
\centering
\includegraphics[width=1.0\textwidth]{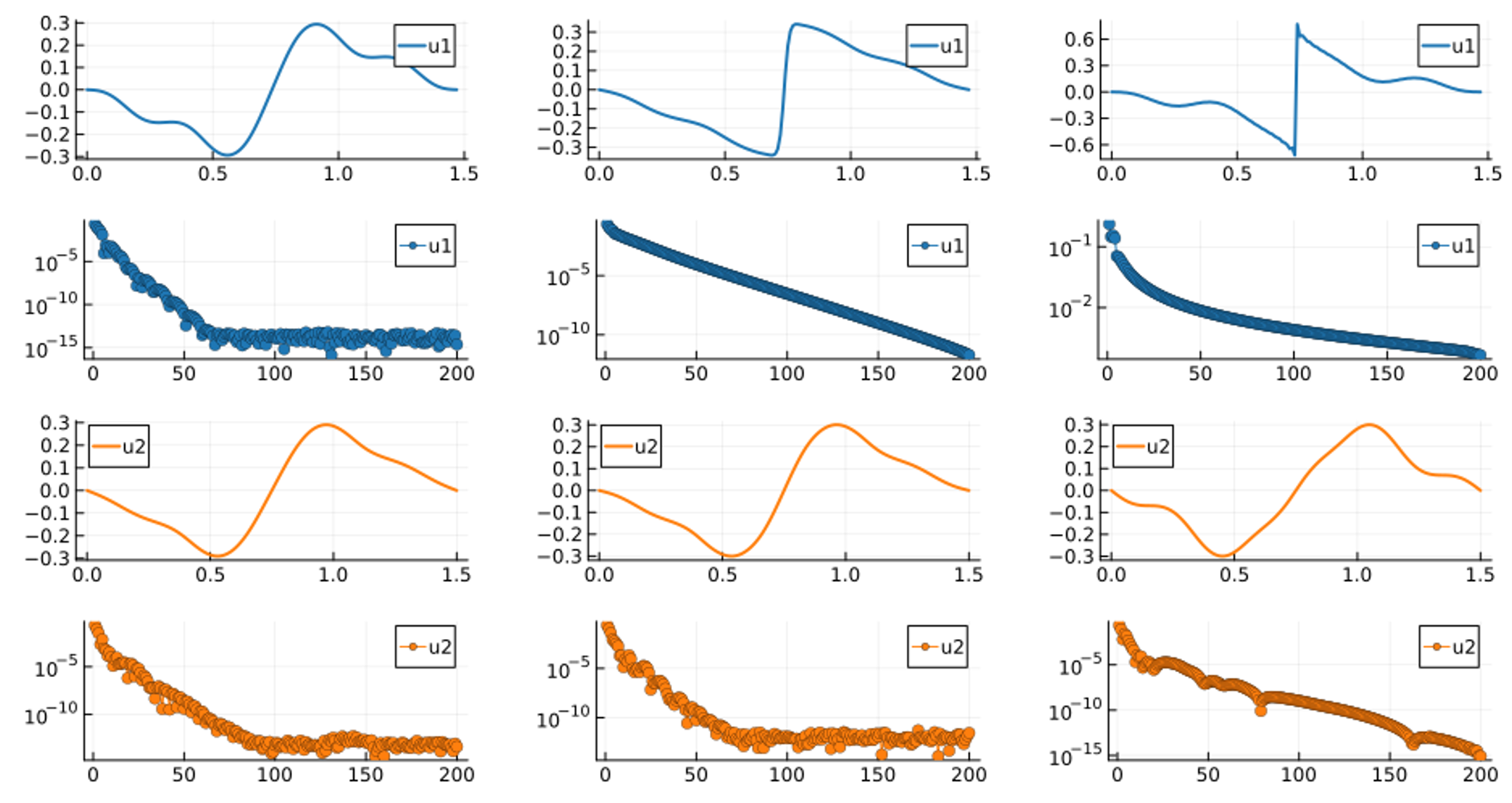}
\put(-380,-10){\makebox(0,0){(c) $\dist=3.6$}}
\put(-230,-10){\makebox(0,0){(b) $\dist=3.2$}}
\put(-80,-10){\makebox(0,0){(a) $\dist=2.8$}}
\caption{Displacements $u_j$ ($j=1,2$) of equilibrium under different $\dist$ (by using scattering-channel cutoffs $L=40$, $W=10$ and $\ucut=200$). The first and third rows display the displacements in configuration space, while the second and fourth rows present their corresponding Fourier coefficients.}
\label{fig:u_compare_dz}
\end{figure}

\vskip 0.2cm
{\bf Convergence w.r.t. scattering-channel cutoffs.}
We then present in Figure~\ref{fig:converge} the convergence of the numerical approximations with respect to the scattering-channel cutoffs $W$, $L$, and $\ucut$.
We measure both the errors in the energy and the $L^2$-error of the displacement field, where the reference solutions are obtained by using sufficiently large cutoffs.

The convergence with respect to $W$ reflects the decay of high physical scattering-channel frequencies. 
This part is primarily controlled by the decay of the observable function $g_{\beta,\mu}$. 
In the Fermi-Dirac case, decreasing $\beta$ leads to a slower decay of $g_{\beta,\mu}$, and the numerical results show a correspondingly slower convergence with respect to $W$. 
This agrees with the theoretical estimate, where the $W$-dependent error is governed by the decay parameter of the observable function.

The convergence with respect to $L$ and $\ucut$ reflects a different mechanism. 
The cutoff $L$ controls the incommensurate direction in the product reciprocal lattice, while $\ucut$ controls the Fourier resolution of the displacement field. 
Their convergence is therefore more closely related to the regularity of the relaxed displacement field than to the high-energy decay of $g_{\beta,\mu}$. 
The numerical results show exponential decay for both errors, and the observed convergence rates remain largely stable across the tested values of $\beta$. 
This is consistent with Theorem \ref{theorem:u:convergence}, where the $L$- and $\ucut$-dependent terms are governed by the regularity exponent of the equilibrium displacement.

\begin{figure}[htbp!]
\centering
\includegraphics[width=0.3\textwidth]{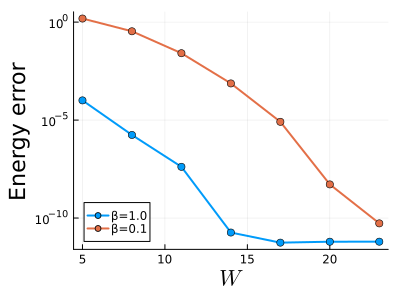}
\includegraphics[width=0.3\textwidth]{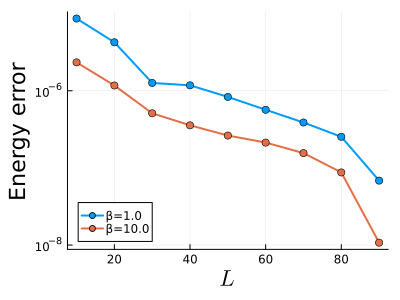}
\includegraphics[width=0.3\textwidth]{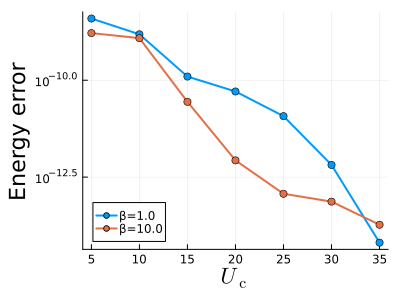}\\
\includegraphics[width=0.3\textwidth]{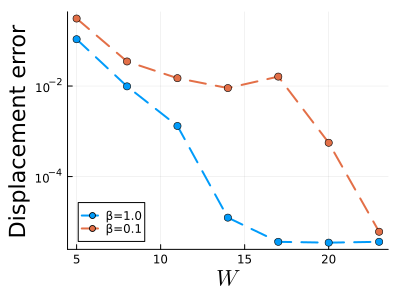}
\includegraphics[width=0.3\textwidth]{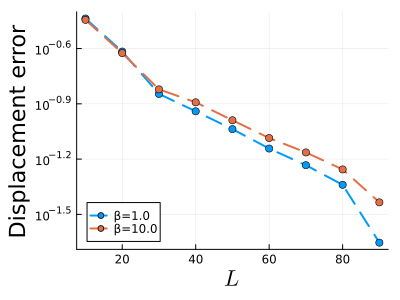}
\includegraphics[width=0.3\textwidth]{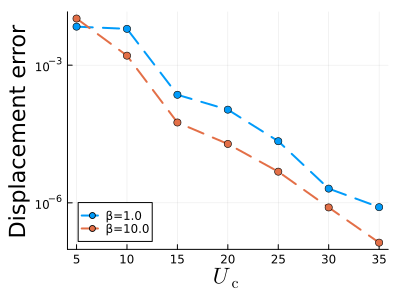}
\caption{Convergence with respect to the truncation parameters $W$, $L$, and $\ucut$.}
\label{fig:converge}
\end{figure}

\vskip 0.2cm
{\bf Simulations of the domain walls.}
We shall then use our model to simulate some nontrivial relaxation patterns, in particular the formation of domain walls in incommensurate bilayers. 
To visualize the relaxed geometry in real space, we plot the relaxed inter-layer shift rather than the displacement field itself. 
For sites $\ell\in\mathcal R_1$, we define the relaxed inter-layer shift by
\begin{align*}
b_1(\ell) + u_1(b_1(\ell)) - u_2(-b_1(\ell)) \qquad{\rm for}~ \ell \in \RL_1.
\end{align*}
This quantity measures the local relative registry between the two layers after relaxation. 
It combines the original local configuration $b_1(\ell)$ with the relaxation of both layers and therefore provides a direct way to identify domains and transition regions in the relaxed structure. 
An analogous expression can be defined for sites in $\RL_2$, which contains the same geometric information and is therefore omitted.

We first study how the relaxed pattern depends on the lattice mismatch in Figure \ref{fig:mismatch_compare}.
The lattice constants and corresponding mismatch parameters are shown in Table \ref{tab:lattice_mismatch}.
After relaxation, the systems exhibit distinct moir\'{e} patterns, where the relaxed shift forms extended plateau regions separated by narrow transition layers. 
The plateau regions correspond to locally preferred stacking configurations, whereas the transition layers represent domain walls between different local registries.
As the mismatch ratio decreases, we observe that the moir\'{e} length scale increases, as each domain occupies a larger spatial region and the domain walls become more widely separated. 
This trend is consistent with the geometric origin of the moir\'{e} pattern, where the moir\'{e} period is inversely related to the lattice mismatch.

\begin{table}[htbp]
    \centering
    \begin{tabular}{ccc}
        \hline
        ~~$L_1$~~ & ~~$L_2$~~ & ~~$|L_1-L_2|/L_1$~~ \\
        \hline
        1.5 & $\sqrt{2}$ & $\approx 0.06$ \\
        \hline
        1.5 & $\sqrt{2}*1.04$ & $\approx 0.02$ \\
        \hline
        1.5 & $\sqrt{2}*1.05$ & $\approx 0.01$ \\
        \hline
    \end{tabular}
    \caption{Lattice constants and approximate mismatch parameters.}
    \label{tab:lattice_mismatch}
\end{table}

We next study the dependence of the relaxed structure on the interlayer coupling strength in Figure~\ref{fig:diff_dz_domain}. 
The interlayer distance $\dist$ controls the strength of the interaction between the two layers and hence the width and sharpness of the domain walls.
We observe from the figure that decreasing $\dist$ makes the transition layers sharper. 
For smaller interlayer distances, the relaxed shift changes over fewer atomic sites, indicating that the equilibrium displacement contains stronger high-frequency components. 
This is the regime in which the regularity assumption on the displacement field becomes more delicate.
By contrast, when the interlayer distance is larger, the relaxed shift varies more smoothly across the configuration space and the resulting domain walls are broader and more regular. 

\begin{figure}[htbp!]
\centering
\includegraphics[width=1.0\textwidth]{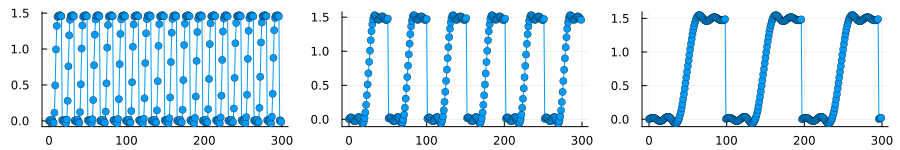}\\[2mm]
\caption{Inter-layer shift of equilibria with different mismatch ratios (by using scattering-channel cutoffs $L=40$, $W=10$ and $\ucut=200$ with $\dist=3.6$). 
Left: $\rm{mismatch}\approx 0.06$. Middle: $\rm{mismatch}\approx 0.02$. Right: $\rm{mismatch}\approx 0.01$. }
\label{fig:mismatch_compare}
\end{figure}

\begin{figure}[htbp!]
\centering
\includegraphics[width=1.0\textwidth]{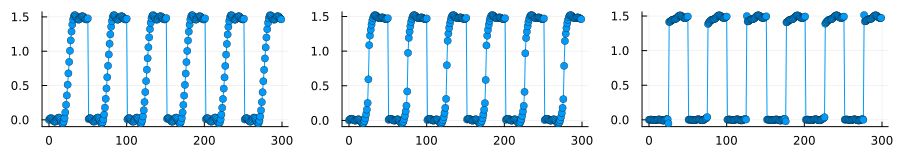}
\caption{Inter-layer shift of equilibria with different vertical distances (by using scattering-channel cutoffs $L=40$, $W=10$ and $\ucut=200$).  
Left: $\dist=2.8$. Middle: $\dist=3.2$. Right: $\dist=3.6$.}
\label{fig:diff_dz_domain}
\end{figure}

\vskip 0.2cm
{\bf Relaxed LDoS in reciprocal space.}
We finally examine how the structural relaxation affects the electronic spectrum, by showing the corresponding changes in the reciprocal-space LDoS defined as
\begin{align*}
    \widehat{\rm{LDOS}}(\xi,g) = \big[g(\hH[u](\xi))\big]_{\bzero,\bzero}.
\end{align*}
This provides a ``band"-like spectral signature of the relaxed geometry. And we use $g$ is a $E$-recentered Gaussian so we get function of $E$.

We first show the comparisons for different interlayer distances in Figure~\ref{fig:ldos_compare}.
For the larger separation ($\dist=3.6$), the interlayer coupling is weak, the LDoS is dominated by nearly parabolic bands associated with the kinetic energy, and the structural relaxation only produces a moderate spectral modification. 
In contrast, for the smaller separation ($\dist=2.0$), the interlayer coupling is much stronger, and the LDoS then exhibits clear band hybridization and pronounced spectral splitting. 
The contrast plots show that relaxation substantially redistributes the spectral weight near the band crossings, indicating the formation of Van Hove type features induced by the moir\'{e} potential and enhanced by structural relaxation \cite{fang2015ab, massatt2017electronic}. 

We further study the role of the shape of the single-site potential \eqref{siteP:numerics}, by comparing two cases: an isotropic potential with $\sigma_z=\sigma_x$ and an anisotropic potential with $\sigma_z=2\sigma_x$.
The results are shown in Figure~\ref{fig:ldos_compare_sigmaz}, where we focus on the band near $\xi=0$ such that kinetic energy is negligible and the moir\'{e} potential dominates. 
We observe that structural relaxation alters the band shapes, leading to the emergence of new energy gaps.
The relaxation effect is more significant when $\sigma_z=2\sigma_x$. 
In this case, the effective interlayer coupling is stronger and more spatially extended, so the structural relaxation produces a larger redistribution of spectral weight. 

\begin{figure}[htbp!]
\centering
\includegraphics[width=1.0\textwidth]{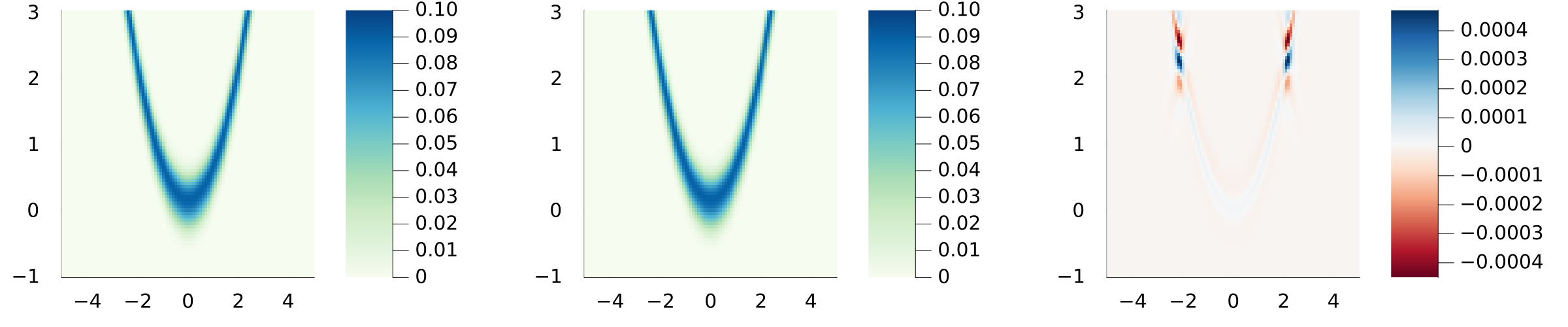}
\includegraphics[width=1.0\textwidth]{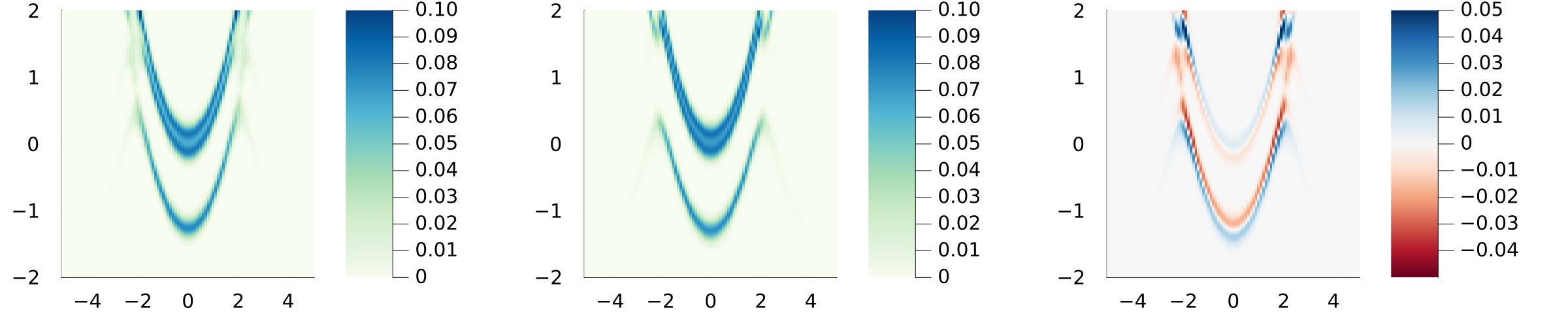}
\caption{LDoS in reciprocal space for $\dist=3.6$ (top row) and $\dist=2.0$ (bottom row). 
Left: unrelaxed LDOS. Middle: relaxed LDOS. Right: difference between the relaxed and unrelaxed LDoS.}
\label{fig:ldos_compare}
\end{figure}

\begin{figure}[htbp!]
\centering
\includegraphics[width=1.0\textwidth]{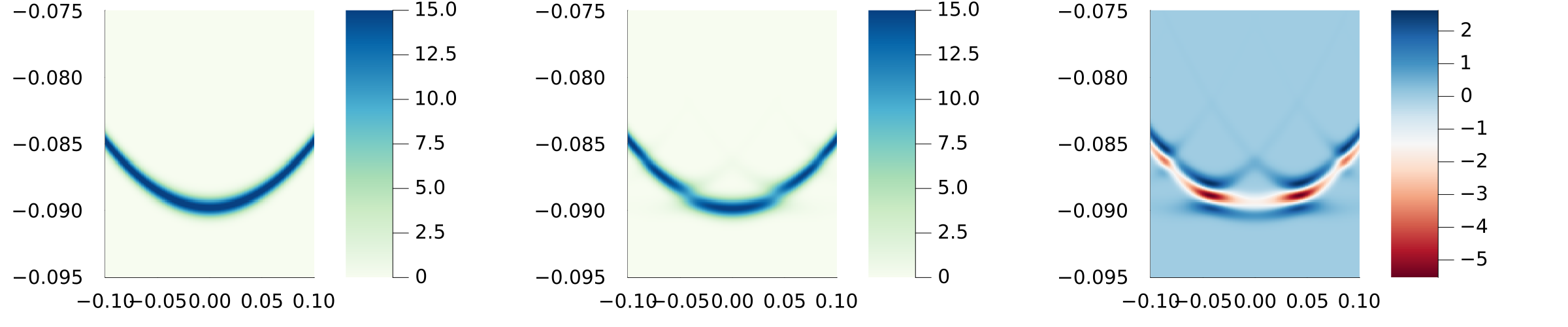}
\includegraphics[width=1.0\textwidth]{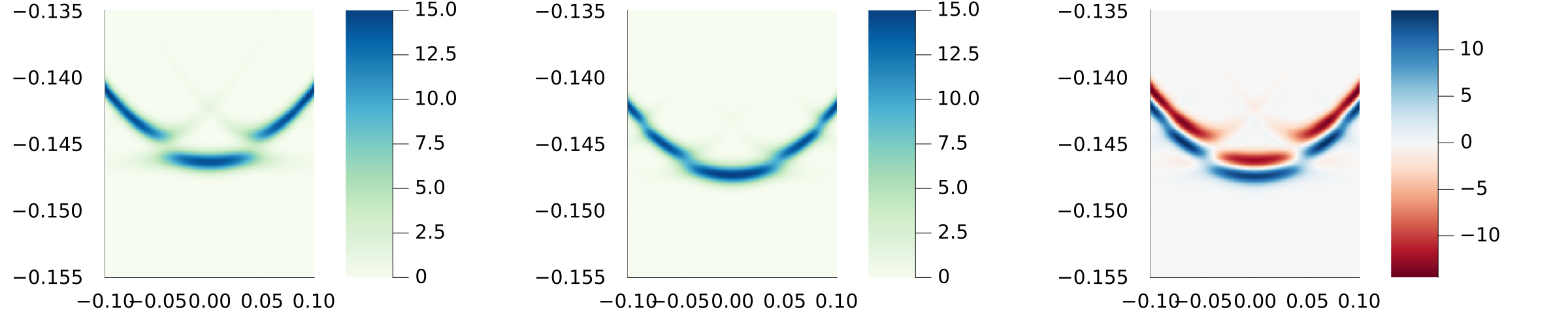}
\caption{LDoS in reciprocal space for $\sigma_z=\sigma_x$ (top row), and $\sigma_z=2\sigma_x$ (bottom row). 
Left: unrelaxed LDOS. Middle: relaxed LDOS. Right: difference between the relaxed and unrelaxed LDoS.}
\label{fig:ldos_compare_sigmaz}
\end{figure}

\section{Conclusions}
\label{sec:conclude}

In this work, we developed a variational quantum framework for studying structural relaxation in incommensurate bilayer systems based on electronic structure models. By formulating the displacement field on the configuration space and representing the continuum Schr\"{o}dinger operator in reciprocal space, we obtained well-defined notions of relaxed energy, local density of states, and forces on the deformed geometry.

We established the thermodynamic limit of the relevant observables and proposed an anisotropic scattering-channel approximation scheme adapted to the reciprocal-space structure of incommensurate systems. Under suitable regularity and stability assumptions, we proved exponential convergence of the approximate equilibria with respect to the scattering-channel cutoffs. Numerical experiments support the theoretical analysis and demonstrate that the model captures both structural relaxation patterns and their effects on the electronic spectrum.

Future work will extend the present framework to more realistic self-consistent electronic structure models. This would make it possible to study important physical effects such as screening, charge redistribution, transport, and topological responses in incommensurate materials.

\appendix
\renewcommand\thesection{\appendixname~\Alph{section}}

\section{Proofs}
\label{sec:proofs}
\renewcommand{\theequation}{A.\arabic{equation}}
\setcounter{equation}{0}
\renewcommand{\thelemma}{A.\arabic{lemma}}
\setcounter{lemma}{0}

Throughout the proofs, we will use the symbol $C$ to denote a generic positive constant that may change from one line to the next. The dependence of $C$ on model parameters will normally be clear from the context or stated explicitly. To simplify the expression, we introduce the resolvent of the operator $A$, denoted by $\resol_z(A):=(z-A)^{-1}$.

\subsection{Proof of Lemma \ref{lemma:potential}}
\label{sec:proof:lemma:potential}

\begin{proof}
We first prove (1).
By restricting $W_j[u_j]$ to the diagonal where $x_{\Fj} = x_j$, we recover the site potential $V_j[u_j](x_j)$ given by
\begin{align*}
W_j[u_j](x_j,x_j)&=\sum_{\ell \in \RL_j}v(x_j-\ell-u_j(x_j-x_j+\ell))
=\sum_{\ell \in \RL_j}v(x_j-\ell-u_j(\ell))
= V_j[u_j](x_j).
\end{align*}
This identity shows that the potential function $V_j[u_j](x_j)$ is the diagonal of the auxiliary potential $W_j[u_j](x_j, x_{\Fj})$.

\vskip 0.2cm

We then show (2).
The $\RL_{\Fj}$-periodicity of $W_j[u_j]$ with respect to $x_{\Fj}$ follows directly from the periodicity of $u_j$. For the $x_j$-direction, let $l' \in \RL_j$. By the translation invariance of the lattice, we have:
\begin{align*}
W_j[u_j](x_j+\ell',x_{\Fj}) &= \sum_{\ell \in \RL_j} v\big(x_j - (\ell-\ell') - u_j((\ell-\ell') + x_{\Fj} - x_j)\big) \\
&= \sum_{\ell'' \in \RL_j} v\big(x_j - \ell'' - u_j(\ell'' + x_{\Fj} - x_j)\big) \\
&= W_j[u_j](x_j,x_{\Fj}),
\end{align*}
where we have used the substitution $\ell'' = \ell - \ell'$ for the second equality. 
Thus, $W_j[u_j](x_j,x_{\Fj})$ is periodic on $\Gamma_j \times \Gamma_{\Fj}$.

\vskip 0.2cm

Finally we shall prove (3). Since $W_j[u_j](x_j, x_{\Fj})$ is periodic on $\Gamma_j \times \Gamma_{\Fj}$, we restrict our analysis to the compact closure of its fundamental domain, denoted by $ \overline{\Omega} \subset \R^d \times \R^d$. Furthermore, since $\overline{\Omega}$ is compact, there exists a uniform constant $R > 0$ such that $|x_j| \le R$ for all $(x_j, x_{\Fj}) \in \overline{\Omega}$.

Since $u_j \in C_{\rm per}(\Gamma_{\Fj}, \R^d)$ is continuous, it is globally bounded by $M>0$. 
For any lattice site $\ell \in \RL_j$, the localized site potential $v_j$ satisfies the reverse triangle inequality:
\begin{equation*}
|x_j - \ell - u_j(x_{\Fj} - x_j + \ell)| \ge |\ell| - |x_j| - |u_j| \ge |\ell| - R - M.
\end{equation*}
Moreover, the uniform upper bound for $v_j\big(x_j-\ell-u_j(x_{\Fj}-x_j+\ell)\big)$ is denoted by  
\begin{align*}
    C_\ell = \sup_{|y| \ge \max(0, |\ell| - R - M)} |v_j(y)|,
\end{align*}
which is independent of the variables $(x_j, x_{\Fj}) \in \overline{\Omega}$. Combined with the exponential decay \eqref{decay:v} of the localized site potential $v_j$, there exists $C>0$ such that
\begin{align*}
    C_\ell \le C e^{-a\max(0, |\ell| - R - M)} \le C' e^{-a|\ell|},
\end{align*}
where $C'$ absorbs the constant factor about $R+M$.
Due to this exponential decay, the series $\sum C_\ell$ is summable.
Consequently, the Weierstrass M-test guarantees that the series converges uniformly on $\overline{\Omega}$. Since each term $v_j\big(x_j-\ell-u_j(x_{\Fj}-x_j+\ell)\big)$ is a composition of continuous functions, their uniform limit $W_j[u_j]$ is continuous on $\overline{\Omega}$. Finally, by periodicity, this continuity extends to the whole space $\R^d \times \R^d$.
\end{proof}

\subsection[Proof of connection]{Proof of \eqref{connection}}
\label{sec:proof:connection}

\begin{proof}
By applying the Fourier transform $\FT$ to the Schr\"{o}dinger operator $\ham[u]$, the geometric periodicity of the auxiliary potentials $W_j$ allows us to express its action as
\begin{align*}
\big(\FT(\ham [u]\phi)\big)(\xi)
&= \int_{\R^{d}} e^{-i\xi \cdot x} \big(-\frac{1}{2}\Delta+\sum_{j=1,2}V_j[u_j](x)\big)\phi(x) \dd x 
\\[1ex]
&=\int_{\R^{d}} e^{-i\xi \cdot x} \big(-\frac{1}{2}\Delta+\sum_{j=1,2}V_j[u_j](x)\big) \frac{1}{(2\pi)^d} \int_{\R^{d}} e^{i\xi' \cdot x} \hat\phi(\xi') \dd \xi' \dd x
\\[1ex]
&= \frac{1}{2}|\xi|^2 \hat\phi(\xi)+\sum_{j=1,2}\sum_{G_1\in \RL^*_1}\sum_{G_2\in \RL^*_2}  \widehat{W_j}[u_j](G_1,G_2)\\[1ex]
&\hspace{5em} \cdot \frac{1}{(2\pi)^{d}}\int_{\R^{d}}\int_{\R^{d}} e^{-i\xi \cdot x} e^{i(G_1\cdot x+G_2\cdot x)} e^{i\xi' \cdot x} \hat\phi(\xi') \dd \xi' \dd x \\[1ex]
&=\frac{1}{2}|\xi|^2 \hat\phi(\xi)+\sum_{j=1,2}\sum_{G_1\in \RL^*_1}\sum_{G_2\in \RL^*_{2}}  \widehat{W_j}[u_j](G_1,G_2)\cdot \hat\phi(\xi-G_1-G_2),
\end{align*}
where $G_1\in \RL^*_1,~G_2\in \RL^*_2$, $u = [u_1,u_2] \in \X$.
This relation naturally defines the discrete Hamiltonian matrix $\hH:\C(\RL_1^* \times \RL_2^*) \rightarrow \C(\RL_1^* \times \RL_2^*)$, whose matrix elements are given by
\begin{equation*}
\hH[u]_{\bG,\bG'} = \frac{1}{2}|G_1+G_2|^2\delta_{G_1G_1'}\delta_{G_2G_2'} + \sum_{j=1,2}\hat{W}_j[u_j](G_1-G_1',G_2-G_2').
\end{equation*}
For ${\bf G} = (G_1,G_2) \in \RL_1^* \times \RL_2^*$, we have from \eqref{hat:W}, the definition of $\unfold_\xi$ and the definition \eqref{Hhat:xi} of $\hH(\xi)$ that
\begin{align*}
& \Big( \unfold_\xi \big(\ham [u]\psi\big) \Big)(\bG)
= \big(\widehat{\ham [u]\psi}\big)(\xi+G_1+G_2) 
= \int_{\R^d} \big(\ham [u]\psi\big)(x) e^{-i(\xi+G_1+G_2)\cdot x} \dd x
\\[1ex]\nonumber
= & \frac{1}{2}\big|\xi+G_1+G_2\big|^2 \hat{\psi}(\xi+G_1+G_2) \\[1ex]
& \hspace{5em} + \sum_{j=1,2}\sum_{G_1'\in \RL^*_1}\sum_{G_2'\in \RL^*_{2}}  \widehat{W_j}[u_j](G_1',G_2')\cdot \hat\psi(\xi-G_1'-G_2'+G_1+G_2)\\[1ex]\nonumber
= & \frac{1}{2}\big|\xi+G_1+G_2\big|^2 \hat{\psi}(\xi+G_1+G_2)\\[1ex]
& \hspace{5em}  + \sum_{j=1,2}\sum_{G_1'\in \RL^*_1}\sum_{G_2'\in \RL^*_{2}} \hat{W}_j[u_j](G_1-G_1',G_2-G_2') \hat{\psi}(\xi+G_1'+G_2')\\[1ex]
=& \Big(\hH[u](\xi) \big(\unfold_\xi \psi\big)\Big)(\bG),
\end{align*}
which hence leads to \eqref{connection}.
\end{proof}

\subsection{Proof of Theorem \ref{theo:TDL} }
\label{sec:proof:theorem:Dos}

To justify the existence of the thermodynamic limit of $\aTr_R\big(g(\ham)\big)$, we first establish the decay of the auxiliary potential $W$ in reciprocal space.
\begin{lemma}
\label{lemma:decay of W}
Let $u\in\X^{\gamma}$, $g \in \setg$. 
Then there exists $\tilde\gamma > 0$ depending on $\gamma$ such that
\begin{equation}
\label{decay:W}
   \Big| \hat{W}_{j}[u_j](\bG) \Big| \le C e^{-\frac{\alpha}{2}|G_1 + G_2| -\tilde \gamma |G_1 - G_2|}
\end{equation}
for any $\bG=(G_1,G_2)\in\RL_1^*\times\RL_2^*$ and $j=1,2$.
\end{lemma}

\begin{proof}
To establish the decay of the Fourier modes, substituting the definition \eqref{W:auxiliary} into \eqref{hat:W} yields
\begin{equation}
\label{W:0}
\hat{W}_{j}[u_j](\bG) = \frac{1}{|\Gamma_j||\Gamma_{\Fj}|} \int_{\Gamma_j}\int_{\Gamma_{\Fj}} \sum_{l \in \RL_j} v(x_j-l-u_j(x_{\Fj}-x_j+l)) e^{-i(G_j\cdot x_j+ G_{\Fj}\cdot x_{\Fj})} \dd x_{\Fj} \dd x_j.
\end{equation}
By exploiting the localization of $v$ and the periodicity with respect to $G_j \in \RL_j^*$, we introduce the change of variables $t = x_j-l$ and $s = (x_{\Fj}-x_j+l) \mod \Gamma_{\Fj}$. Since the Jacobian of this mapping is unity, extending the integration over $t$ to $\R^d$ transforms \eqref{W:0} into
\begin{equation*}
\hat{W}_{j}[u_j](\bG) = \frac{1}{|\Gamma_j||\Gamma_{\Fj}|} \int_{\Gamma_{\Fj}} \int_{\R^d} v(t-u_j(s)) e^{-i(G_1+G_2)\cdot t} \dd t ~ e^{-i G_{\Fj}\cdot s} \dd s.
\end{equation*}
Setting $w = G_1+G_2$ and applying the shift $\tau = t - u_j(s)$ to the inner integral yields
\begin{equation*}
\int_{\R^d} v(t-u_j(s)) e^{-i w\cdot t} \dd t = e^{-i w\cdot u_j(s)} \int_{\R^d} v(\tau) e^{-i w\cdot \tau} \dd\tau = e^{-i w\cdot u_j(s)} \hat v(w),
\end{equation*}
where $\hat v(w)$ satisfies $|\hat v(w)| \le C e^{-\alpha |w|}$. Substituting this inner integral back directly leads to the factorized form
\begin{equation}
\label{form:W}
\hat{W}_{j}[u_j](\bG) = \frac{\hat v(G_1+G_2)}{|\Gamma_j|} \cdot \hat b[G_1+G_2, u_j](G_{\Fj})
\end{equation}
with the displacement factor $\hat b$ defined by
\begin{equation}
\label{form:b}
    \hat b[w, u_j](G) = \frac{1}{|\Gamma_{\Fj}|} \int_{\Gamma_{\Fj}} e^{-i w\cdot u_j(s)} e^{-i G\cdot s} \dd s.
\end{equation}
For the decay of $\hat{b}$, we note that the displacement field $u_j(s) \in \X^\gamma$ can be analytically extended to the strip $\{z \in \C^d:|\operatorname{Im}(z)| \le \rho\}$ for any $\rho < \gamma$. Since the integrand $e^{-iw\cdot u_j(z)}e^{-iG_{\Fj}\cdot z}$ is analytic in the strip and periodic with respect to $\mathcal R_{\Fj}$ in the real directions, the integral over the torus is invariant under translations in the imaginary direction inside the strip.
Thus, for any $|\eta|\le \rho$, shifting the contour from the real torus to the translated torus $\Gamma_{\Fj}+i\eta$ gives
\begin{align}
\label{shift:b}
    \hat b[w,u_j](G_{\Fj}) = \frac{1}{|\Gamma_{\Fj}|} \int_{\Gamma_{\Fj}} e^{-iw\cdot u_j(s+i\eta)} e^{-iG_{\Fj}\cdot (s+i\eta)} \dd s .
\end{align}
To estimate \eqref{shift:b}, we define a bounding function for $0<\rho<\gamma$ defined by
\begin{align*}
    M_j(\rho):=\sum_{Q\in \mathcal R_{\Fj}^*} |\hat u_{Q}|\big(e^{\rho |Q|}-1\big).
\end{align*}
This definition guarantees the pointwise bound
\begin{align*}
    |\operatorname{Im} u_j(s+i\eta)| \le |\operatorname{Im} \big( u_j(s+i\eta)-u_j(s)\big) | \le |u_j(s+i\eta)-u_j(s)| \le M_j(\rho).
\end{align*}
By choosing $\eta=-\rho \frac{G_{\Fj}}{|G_{\Fj}|}$ for $G_{\Fj}\neq 0$, we can bound the displacement factor as
\begin{align}
\label{decay:b}
    |\hat b[w,u_j](G_{\Fj})| &\le \frac{1}{|\Gamma_{\Fj}|} \int_{\Gamma_{\Fj}} \big|  e^{-iw\cdot u_j(s+i\eta)} \big| \big|e^{-iG_{\Fj} \cdot (s+i\eta)} \big| \dd s \nonumber\\
    & \le e^{|w|\,|\operatorname{Im}u_j(s+i\eta)|} e^{G_{\Fj}\cdot\eta}
    \le e^{M_j(\rho)|w|}e^{-\rho |G_{\Fj}|}.
\end{align}
The same estimate also holds when $G_{\Fj}=0$.
Combining \eqref{decay:b} and \eqref{decay:v} with the triangle inequality $|G_1 - G_2| \le |G_1 + G_2| + 2|G_{\Fj}|$ yields the final decay bound
\begin{align}
\label{decay_W:G1-G2}
\nonumber
|\hat{W}_{j}[u_j](\bG)| &\le C e^{-\big(\alpha - M_j(\rho) \big)|G_1+G_2|} e^{-\rho |G_{\Fj}|} \\
&\le C e^{-(\alpha - M_j(\rho) - \rho/2)|G_1+G_2|} e^{-\frac{\rho}{2}|G_1-G_2|}.
\end{align}

To obtain an explicit estimate for $ M_j(\rho)$, we introduce $\displaystyle U^{\gamma} = \max_{\Fj=1,2}\sum_{Q\in \mathcal R_{\Fj}^*} |\hat u_{Q}|e^{\frac{\gamma}{2} |Q|}$.
Since $u \in \X^\gamma$, $U^\gamma$ is uniformly bounded. Then, for any $0 < \rho < \frac{\gamma}{2}$, we can bound $M_j(\rho)$ as 
\begin{align*}
    M_j(\rho) \le \sum_{Q\in \mathcal R_{\Fj}^*} |\hat u_{Q}|\Big(e^{\frac{2\rho}{\gamma} \frac{\gamma}{2} |Q|}-1\Big) \le \sum_{Q\in \mathcal R_{\Fj}^*} |\hat u_{Q}| \frac{2\rho}{\gamma} e^{\frac{\gamma}{2} |Q|} \le \frac{2\rho}{\gamma} U^{\gamma}.
\end{align*}
Consequently, by setting the parameter $\rho_0 = \min\left\{\frac{\gamma}{2}, \frac{2\alpha\gamma}{4U^{\gamma}+\gamma}\right\}$, we explicitly guarantee that
\begin{align*}
    \alpha - M_j(\rho_0) - \rho_0/2 \ge 0.
\end{align*}

To circumvent the excessive loss of decay in the $G_1+G_2$ direction, we observe from \eqref{form:b} the trivial bound $|\hat b[w,u_j](G)| \le 1$. Combining this with the decay of the atomic potential \eqref{decay:v} directly yields
\begin{align}
\label{decay_W:G1+G2}
|\hat{W}_j[u_j](\bG)| &\le C e^{-\alpha |G_1+G_2|}.
\end{align}
Interpolating between the bounds \eqref{decay_W:G1-G2} and \eqref{decay_W:G1+G2} provides the refined inequality
\begin{align*}
|\hat{W}_{j}[u_j](\bG)| &\le C e^{-\frac{\alpha}{2} |G_1+G_2|} e^{-\frac{\rho_0}{4}|G_1-G_2|}.
\end{align*}
The corresponding decay rate with respect to $|G_1-G_2|$ is then denoted by $\tilde{\gamma} = \rho_0/4$. 
\end{proof}

The above estimate shows that the decay rates are different in the directions $|G_1 + G_2|$ and $|G_1 - G_2|$. 
To describe this difference more clearly, we introduce the following seminorms for $\bG = (G_1, G_2) \in \RL_1^* \times \RL_2^*$,
\begin{equation*}
    |\bG|_+ := |G_1 + G_2|
    \qquad{\rm and}\qquad
    |\bG|_- := |G_1 - G_2|.
\end{equation*}
It is straightforward to verify that both of these mappings are well-defined seminorms. 
Indeed, the triangle inequality follows directly from the linearity of each component and the standard properties of the absolute value
\begin{equation*}
    |\bG+\bG'|_+ = |(G_1 + G'_1) + (G_2 + G'_2)| \le |G_1 + G_2 |+ | G'_1 + G'_2| =|\bG|_+ + |\bG'|_+,
\end{equation*}
\begin{equation*}
    |\bG+\bG'|_- = |(G_1 + G'_1) - (G_2 + G'_2)| \le |G_1 - G_2 | + | G'_1 - G'_2| =|\bG|_- + |\bG'|_-.
\end{equation*}
Then we are ready to prove Theorem \ref{theo:TDL}.

\vskip 0.2cm

\begin{proof}[Proof of Theorem \ref{theo:TDL}]
Using the decay of $\Big| \hat{W}_{j}[u_j](\bG) \Big|$ in Lemma \ref{lemma:decay of W} and the relation
\begin{equation*}
    |\bG|_+ + |\bG|_- \ge \sqrt{2}|\bG|,
\end{equation*}
 we obtain the following exponential decay for the off-diagonal elements for $\bG \neq \bG'$,
\begin{equation}
\label{decay:H}
|\hH[u](\xi)_{\bG, \bG'}| \le C e^{-\frac{1}{\sqrt{2}}\tilde\gamma |\bG-\bG'|}.
\end{equation}

Define a truncation set $\Omega_R:= B_R\cap (\RL_1^*\times \RL_2^*)$ ,and $R>0$ is large enough such that $\max\big\{|\bG|,|\bG'|\big\}<R/2$. 
Let $\Omega \subset \RL^*_1\times\RL^*_2$ be a finite set such that $\Omega \supsetneq \Omega_R$.
Then $\hH({\xi})^{\Omega_R}$ is expanded as a bigger matrix $\hH({\xi})^{\Omega_R}_\Omega \in \C^{\Omega\times\Omega}$
\begin{eqnarray*}
\big(\hH[u]^{\Omega_R}(\xi)_\Omega\big)_{\bG',\bG''}=
\left\{
\begin{array}{ll}
\big(\hH[u]^{\Omega_R}(\xi)\big)_{\bG',\bG''} \quad &{\rm if \ |\bG'|,|\bG''| \le R},
\\[1ex]
\quad 0 & {\rm otherwise}.
\end{array}
\right.
\end{eqnarray*}
Then we have $\df\big(\hH[u]^{\Omega_R}_\Omega(\xi)\big) = \df\big(\hH[u]^{\Omega_R}(\xi)\big)\cup\{0\}$, where $\df(\cdot)$ represents the spectrum set of an operator.

Let $\mathfrak{s}(g)$ represent the non-analytic region of the given function $g\in\setg$.
We can find a contour $\Cc_{R}\subset\C$ such that it encloses all the eigenvalues of $\hH[u]^\Omega(\xi)$ and $\hH[u]^{\Omega_R}_\Omega(\xi)$,
and satisfies
\begin{align}
\label{resolvant-dist-R}
\min\Bigg\{
{\rm dist}\big(z,\mathfrak{s}(g)\big), ~
{\rm dist}\left(z,\df\big(\hH[u]^{\Omega_R}_\Omega(\xi)\big)\right), ~
{\rm dist}\left(z,\df\big(\hH[u]^\Omega(\xi)\big)\right), ~
{\rm dist}\big(z, \{0\}\big) \Bigg\} \geq \frac{\delta}{2}
\end{align}
for any $z\in\Cc_{R}$.

Applying a Combes–Thomas type estimate \cite{wang2025convergence,chen2016qm,combes1973asymptotic} yields a constant $c > 0$ independent of $\gamma$ such that for any $z \in \Cc_{R}$, 
\begin{eqnarray}
\label{decay:resol}
\left\vert \resol_z\big(\hH[u]^{\Omega_R}_\Omega(\xi)\big)_{\bG,\bG'} \right\vert \leq C \delta^{-1} e^{-c\Dres |\bG-\bG'|},
\end{eqnarray}
where $\Dres := \min{\{ \tilde\gamma , \delta\}}$. Note that the above estimate  holds for arbitrary $\Omega \supset \Omega_R$. Due to the exponential decay of $g$ given in \eqref{deftestfunctiong}, the integral $ \int_{\Cc_R} |g(z)| \dd z$ admits a uniform bound as $R \to \infty$, which depends only on $\zeta$. Therefore, by \cite[Lemma A.1]{wang2025convergence}, we deduce the limit $\displaystyle \lodoss{\hH(\xi)}{\bG}{\bG'} := \lim_{R \to \infty}\lodoss{\hH(\xi)^{\Omega_{R}}}{\bG}{\bG'}$ exists.
Moreover, there exist constants $C,c>0$, such that
\begin{eqnarray}
\label{locality:reciprocal_R}
\left\vert \lodoss{\hH(\xi)^{\Omega_{R}}}{\bG}{\bG'} - \lodoss{\hH(\xi)}{\bG}{\bG'} \right\vert \leq C \delta^{-2} e^{-c\Dres R} .
\end{eqnarray}
Then applying \cite[Lemma A.2]{wang2025convergence}, we deduce that the matrix elements of $g\big(\hH[u](\xi)\big)$ satisfy
\begin{align}
\label{decay-xi-G1-G2}
\left| g\big(\hH[u](\xi)\big)_{\bG,\bG'} \right|
\le C e^{- c\left( \min \big\{|\xi+G_1+G_2|,|\xi+G'_1+G'_2| \big\} + |\bG-\bG'| \right)} ,
\end{align}
where $C,c > 0$ depend on $\zeta$, $\gamma$, and $\delta$.

With the decay estimates \eqref{locality:reciprocal_R} and \eqref{decay-xi-G1-G2} verifying the assumptions of \cite[Appendix B.2]{wang2025convergence}, we immediately obtain the existence of the thermodynamic limit
\begin{equation*}
    \aTr\big(g(H)\big) := \lim_{R\rightarrow\infty}\aTr_R\big(g(H)\big) = \int_{\R^d} [g(\hH(\xi))]_{{\bf 0,0}} \dd\xi.
\end{equation*}
This completes the proof.
\end{proof}

\subsection{Proof of Theorem \ref{theorem:force}}
\label{sec:proof:force}

We aim to derive the variational derivative $\nabla_u E[u]$ by utilizing the contour integral representation of the energy functional. 
Given the spectral properties of the Hamiltonian, the energy can be expressed via the Cauchy integral formula
\begin{equation*}
E[u]= \int_{\R^d} \frac{1}{2\pi \im} \oint_{\mathscr{C}} g(z) [(z-\hH[u](\xi)^{-1}]_{\bzero,\bzero} \dd z \dd \xi.
\end{equation*}
The functional derivative with respect to the displacement field $u$ is given by
 \begin{align}
 \label{form:F}
\nabla_u E[u]
= \int_{\R^d} \frac{1}{2\pi \im} \oint_{\mathscr{C}} g(z) [\resol_z\big(\hH[u](\xi)\big) \nabla_{u} \hH[u](\xi)  \resol_z\big(\hH[u](\xi)\big) ]_{\bzero,\bzero} \dd z \dd \xi .
\end{align}

To facilitate the calculation, let $\hat u_{Q}$ denote the Fourier coefficients of the displacement field $u_j$ with $Q \in \RL^*_{\Fj}$. 
It follows from \eqref{Hhat:xi} that the partial derivatives of the Hamiltonian $\partial_{\hat u_{Q}} \hat{H}[u](\xi)$ are independent of $\xi$ and take the form
\begin{align}
\label{dH:0}
\partial_{\hat u_{Q}} \hH[u]_{\bG,\bG'}
= \partial_{\hat u_{Q}} \hat W_{j}(\bG-\bG'), ~\rm{for}~ Q\in \RL^*_{\Fj}.
\end{align}
The following lemma shows that we can transform the integral over $\xi$ into the limit of a trace. To simplify the notation, we let
\begin{align*}
    \Fk(\xi) := \partial_{\hat u_{Q}} g(\hH[u](\xi)) = \frac{1}{2\pi \im} \oint_{\mathscr{C}} g(z)\resol_z\big(\hH[u](\xi)\big)  \partial_{\hat u_{Q}} \hH[u] \resol_z\big(\hH[u](\xi)\big)\dd z,
\end{align*}
and
\begin{align*}
    \Fk^{R}(\xi) := \frac{1}{2\pi \im} \oint_{\mathscr{C}_R} g(z)\resol_z\big(\hH[u]^{\Omega_R}(\xi)\big)  \partial_{\hat u_{Q}} \hH[u]^{\Omega_R} \resol_z\big(\hH[u]^{\Omega_R}(\xi)\big) \dd z.
\end{align*}

The proof proceeds in two steps. First, we show that $\Fk(\xi)$ is the thermodynamic limit of its finite-dimensional truncations, which is the content of Lemma \ref{lemma:decay:Fk}. Second, we replace the reciprocal-space integral $\Fk(\xi)_{\bzero,\bzero}$ by the normalized trace of the finite truncations. This averaging identity is stated in Lemma \ref{lemma:integral of AB}. 
Once these two facts are established, the proof of Theorem 2.2 follows by applying the cyclicity of the trace at the finite-dimensional level and then passing to the thermodynamic limit.

The next lemma gives the elements of $\Fk$ as thermodynamic limit in the reciprocal space and shows the decay of $\Fk(\xi)_{\bzero,\bzero}$ with respect to $\xi$ and $|Q|$ which is needed for the trace averaging argument.
Specifically, we define the truncated domain $\Omega_R:= B_R\cap (\RL_1^*\times \RL_2^*)$.
\begin{lemma}[Thermodynamic limit and decay of $\Fk$]
\label{lemma:decay:Fk}
Under the assumptions of Theorem \ref{theorem:force}, Let $R > 0$ be large enough such that  $\max{\{|\bG|,|\bG'|\}}  < R/2$. Then the limit
\begin{align*}
    \Fk(\xi)_{\bG,\bG'} = \lim_{R \rightarrow \infty } \Fk^{R}(\xi)_{\bG,\bG'}
\end{align*}
exists. Moreover, there exist positive constants $C$ dependent on $\delta$ and $c$ independent of $\gamma$ such that
\begin{align}
\label{Fk:converge_R}
    |\Fk(\xi)_{\bG,\bG'} -\Fk^{R}(\xi)_{\bG,\bG'}| \le C e^{-c\Dres \big|R  -\max{\{|\bG|,|\bG'|\}}-\sqrt{2}|Q| \big|}.
\end{align}
where $\Dres= \min{\{\tilde\gamma , \delta\}}$ with $\tilde\gamma>0$ depending on $\gamma$ (as established in Lemma \ref{lemma:decay of W}).
Furthermore, we have the estimate
\begin{align}
\label{Fk:converge_G_xi}
    |\Fk(\xi)_{\bzero,\bzero}| \le C e^{- c\Dres |\xi| - c \Dres |Q| },
\end{align}
with the positive constants $C$ depending on $\zeta, \delta$ depending on $\zeta$, and $c$ depending on $\zeta$. 
\end{lemma}

\begin{proof}
Using \eqref{dH:0}, \eqref{form:W}, and \eqref{form:b}, we obtain the explicit form of $\partial_{\hat u_{Q}} \hH[u]$. Without loss of generality, we may assume that $Q \in \RL^*_{1}$ throughout the subsequent discussion. Letting $\Delta \bG := \bG - \bG'$, we have
\begin{align}
\label{form:dH}
    \partial_{\hat u_{Q}} \hH[u]_{\bG,\bG'} = -i(\Delta G_1+ \Delta G_2)\frac{\hat v(\Delta G_1+ \Delta G_2)}{|\Gamma_1|} 
    \hat b[\Delta G_1+ \Delta G_2, u_j](\Delta G_{1}-Q) .
\end{align}
Since frequency shifts preserve the asymptotic decay rate, substituting $G_1 = \Delta G_1 - Q$ and $G_2 = \Delta G_2 + Q$ into the decay estimates \eqref{decay:b} and \eqref{decay:v} in Lemma \ref{lemma:decay of W} directly yields the bound
\begin{align}
 \label{decay:Hk}
    |\partial_{\hat u_{Q}} \hH[u]_{\bG,\bG'}| 
    &\le C e^{-\frac{\alpha}{2} |\Delta G_1+ \Delta G_2|} e^{-\tilde\gamma | \Delta G_1 - \Delta G_2 -2Q|}.
\end{align}
Combining the proof of Theorem \ref{theo:TDL} with the Combes–Thomas estimate for the resolvents \eqref{decay:resol} directly yields the bound
\begin{align}
\label{decay:Fk-Gi-Gj}
    & ~\Big| \big[|\resol_z\big(\hH[u](\xi)\big)  \partial_{\hat u_{Q}} \hH[u] \resol_z\big(\hH[u](\xi)\big) \big]_{\bG,\bG'} \Big| \nonumber \\[1ex]
    \le & \sum_{\bG'',\bG'''} |\resol_z\big(\hH[u](\xi)\big) _{\bG,\bG''}|\cdot|\partial_{\hat u_{Q}} \hH[u]_{\bG'',\bG'''}| \cdot|\resol_z\big(\hH[u](\xi)\big) _{\bG''',\bG'}| \nonumber \\ 
    \le & C \sum_{\bG'',\bG'''} e^{-\Dres |\bG-\bG''|}  e^{-\frac{\alpha}{2}|(\Delta G_1-Q)+ (\Delta G_2+Q)|} e^{-\tilde\gamma |(\Delta G_1- Q)-(\Delta G_2+ Q)|} e^{-\Dres |\bG'''-\bG'|} \nonumber\\
    \le & C e^{ -c \Dres |\bG-\bG'-(Q,-Q)|},
\end{align}
where we have denoted $\Delta \bG := \bG''-\bG'''$.

To estimate the truncation error, we first observe the resolvent difference expanded as
\begin{align*}
    &\Big| \big(\resol_z\big(\hH[u]^{\Omega_R}(\xi)\big)  -\resol_z\big(\hH[u](\xi)\big)\big)_{\bG,\bG'} \Big| \\
    \le & \sum_{\bG'',\bG'''} \Big|\resol_z\big(\hH[u]^{\Omega_R}(\xi)\big)_{\bG,\bG''}(\hH[u](\xi)-\hH[u]^{\Omega_R}(\xi))_{\bG'',\bG'''}\resol_z\big(\hH[u](\xi)\big) _{\bG''',\bG'} \Big|.
\end{align*}
Applying the decay estimates \eqref{decay:H}, \eqref{decay:resol}, and \eqref{decay:Hk} to bound each term leads to the composite error estimate
\begin{align}
\label{Fk_R}
    &\bigg| \Big[\resol_z\big(\hH[u]^{\Omega_R}(\xi)\big)  \partial_{\hat u_{Q}} \hH[u]^{\Omega_R} \resol_z\big(\hH[u]^{\Omega_R}(\xi)\big) 
    -\resol_z\big(\hH[u](\xi)\big)  \partial_{\hat u_{Q}} \hH[u] \resol_z\big(\hH[u](\xi)\big) \Big]_{\bG,\bG'} \bigg| \nonumber \\[1ex]
    \le&  \sum_{\bG'',\bG'''} \bigg( \Big| \big[\resol_z\big(\hH[u]^{\Omega_R}(\xi)\big)  -\resol_z\big(\hH[u](\xi)\big)\big]_{\bG,\bG''} \Big|
     \cdot  \Big|\partial_{\hat u_{Q}} \hH[u]^{\Omega_R}_{\bG'',\bG'''}\resol_z\big(\hH[u]^{\Omega_R}(\xi)\big) _{\bG''',\bG'} \Big| \nonumber\\[1ex]
    &+  \Big|\resol_z\big(\hH[u](\xi)\big) _{\bG,\bG''} \Big|\cdot \Big| \big[\partial_{\hat u_{Q}} \hH[u]^{\Omega_R} -\partial_{\hat u_{Q}} \hH[u] \big]_{\bG'',\bG'''} \Big|\cdot  \Big|\resol_z\big(\hH[u]^{\Omega_R}(\xi)\big) _{\bG''',\bG'} \Big| \nonumber\\[1ex]
    &+ \Big|\resol_z\big(\hH[u](\xi)\big)_{\bG,\bG''} \partial_{\hat u_{Q}} \hH[u]_{\bG'',\bG'''} \Big| \cdot 
    \Big| \big[\resol_z\big(\hH[u]^{\Omega_R}(\xi)\big)  -\resol_z\big(\hH[u](\xi)\big)\big]_{\bG''',\bG'} \Big| \bigg) \nonumber\\[1ex]
    \le&  C(2\delta^{-3} +\delta^{-2}) e^{-c\Dres \big|R -\max{\{|\bG|,|\bG'|\}}-\sqrt{2}|Q| \big|}
\end{align}
for any $\bG,\bG' \in \Omega_{R}$. 

Combining the above estimates with the techniques from the proof of Theorem \ref{theo:TDL} (specifically, the approach of \cite[Lemma A.1]{wang2025convergence}), and leveraging the exponential decay of $g$ as defined in \eqref{deftestfunctiong}, we deduce that the integral $\oint_{\mathscr{C}_R} |g(z)| \dd z$ is uniformly bounded by a constant depending on $\zeta$ as $R \rightarrow \infty$. Consequently, the limit of $A^R(\xi)_{\bG,\bG'}$ is well-defined, which establishes the estimate \eqref{Fk:converge_R}. Thus, for $\bG,\bG' \in \Omega_{R/2}$, the limit $\Fk(\xi)_{\bG,\bG'} = \lim_{R \rightarrow \infty } \Fk^{R}(\xi)_{\bG,\bG'}$ exists.

Moreover, following an approach similar to Ger\v{s}gorin's theorem as adapted in \cite[Lemma A.2]{wang2025convergence}, we obtain the localization of $\Fk(\xi)$.
For $\xi \in \R^d$, we let $\Omega_{\bar{R}}:=B_{\bar{R}} \cap \big( \RL_1^* \times \RL_2^* \big)$ and $\bar{R}:=\frac{1}{2}|\xi|$.
By using the decay of $|\hat{W}_{j}[u_j](\bG)|$ in \eqref{decay:W}, we see that there exists a constant $C_W>0$ such that $\sum_{j=1}^2 \sum_{\bG\in \RL_1^* \times \RL_2^*}|\hat{W}_{j}[u_j](\bG)|<C_W$. 
This, together with the definition \eqref{Hhat:xi} and Ger\v{s}gorin's theorem \cite{horn2012matrix}, implies that
\begin{align*}
\df\big(\hH(\xi)^{\Omega_{\bar{R}}}\big) \subset \bigcup_{(G''_1,G''_2)\in \Omega_{\bar{R}}} 
\Bigg\{\lambda\in \R: ~
\left\vert \lambda - \frac{1}{2}\left\vert \xi + G''_1+G''_2 \right\vert^2 \right\vert \leq C_W \Bigg\} .
\end{align*}
Therefore, for $g\in\setg$, there exists a contour $\Cc$ that encloses all the eigenvalues of $\hH(\xi)^{\Omega_{\bar{R}}}$ and satisfies
\begin{align}
\nonumber
& \min\Big\{ {\rm dist}\big(z,\mathfrak{s}(g)\big) ,~ {\rm dist}\big(z,\df\big(\hH(\xi)^{\Omega_{\bar{R}}}\big) \Big\} > \frac{\delta}{2}
\qquad\qquad{\rm and}
\\[1ex]
\label{proof:est:Rez}
& \big| {\rm Re}(z) \big| \geq \min\limits_{(G''_1,G''_2)\in\Omega_{\bar{R}}} \Bigg\{\frac{1}{2}|\xi+G''_1+G''_2|^2 \Bigg\} - C_W - 1
\end{align}
for any $z\in\Cc$.
We observe from $\bar{R}=\frac{1}{2}|\xi|$ that
\begin{equation*}
\min\limits_{(G''_1,G''_2)\in\Omega_{\bar{R}}} \frac{1}{2}|\xi + G_1''+G_2''|^2 \geq \frac{1}{2} \bar{R}^2 = \frac{1}{8}|\xi|^2.
\end{equation*}
This together with \eqref{proof:est:Rez} implies
\begin{align}
\label{proof-Rez-estimate}
|{\rm Re}(z)| \geq \frac{1}{8}|\xi|^2 - C_W - 1 \geq \frac{1}{8}|\xi| - \bar{C}_W
\qquad\forall~z\in\Cc ,
\end{align}
where $\bar{C}_W$ is a constant depending only on $C_W$.
Then by combining the decay property of $g$ (indicated in \eqref{deftestfunctiong}), the estimate of resolvent as that of \eqref{decay:Fk-Gi-Gj}, the contour integral representation and \eqref{proof-Rez-estimate}, we have 
\begin{align}
\label{proof-decay-dgH-xi}
&\left\vert  \frac{1}{2\pi \im} \oint_{\mathscr{C}} g(z) [\resol_z\big(\hH[u]^{\Omega_{\bar{R}}}(\xi)\big) \partial_{\hat u_{Q}} \hH[u]^{\Omega_{\bar{R}}}  \resol_z\big(\hH[u]^{\Omega_{\bar{R}}}(\xi)\big) ]_{\bzero,\bzero} \dd z \right\vert \nonumber\\
& \leq \frac{1}{2\pi} \oint_{\Cc} |g(z)| \cdot \left\vert [\resol_z\big(\hH[u]^{\Omega_{\bar{R}}}(\xi)\big) \partial_{\hat u_{Q}} \hH[u]^{\Omega_{\bar{R}}}  \resol_z\big(\hH[u]^{\Omega_{\bar{R}}}(\xi)\big) ]_{\bzero,\bzero} \right\vert \dd z \nonumber
\\[1ex]
& \leq C \frac{1}{2\pi} \oint_{\Cc} e^{-\zeta\big(\frac{1}{8}|\xi| - \bar{C}_W\big)} e^{-c \Dres |Q|} \dd z \nonumber
\\[1ex] 
& \leq C e^{-c\zeta |\xi|}  e^{-c \Dres |Q|} .
\end{align}
By the estimate \eqref{Fk_R}, we can immediately deduce that
\begin{align}
\label{diff_bound}
&\Big| \big[\resol_z\big(\hH[u]^{\Omega_{\bar{R}}}(\xi)\big) \partial_{\hat u_{Q}} \hH[u]^{\Omega_{\bar{R}}}  \resol_z\big(\hH[u]^{\Omega_{\bar{R}}}(\xi)\big) \big]_{\bzero,\bzero} \nonumber \\
&\hspace{5em}
-\big[\resol_z\big(\hH[u](\xi)\big) \partial_{\hat u_{Q}} \hH[u]  \resol_z\big(\hH[u](\xi)\big) \big]_{\bzero,\bzero} \Big|\le C e^{-c\Dres\big|\bar{R}-\sqrt{2}|Q|\big|}.
\qquad
\end{align}
Furthermore, \eqref{decay:Fk-Gi-Gj} also implies the following individual bounds:
\begin{align}
\label{indiv_bound}
    [\resol_z\big(\hH[u]^{\Omega_{\bar{R}}}(\xi)\big) \partial_{\hat u_{Q}} \hH[u]^{\Omega_{\bar{R}}}  \resol_z\big(\hH[u]^{\Omega_{\bar{R}}}(\xi)\big) ]_{\bzero,\bzero} \le C e^{-c \Dres |Q|},\\
   [\resol_z\big(\hH[u](\xi)\big) \partial_{\hat u_{Q}} \hH[u] \resol_z\big(\hH[u](\xi)\big) ]_{\bzero,\bzero} \le C e^{-c \Dres |Q|}.
\end{align}
By considering the relative magnitudes of $\bar{R}$ and $|Q|$, we establish the global bound. 
For $\bar{R} \ge 2\sqrt{2}|Q|$, we have $\bar{R} - \sqrt{2}|Q| \ge \bar{R}/2$. Thus, the difference is bounded by $\min\big\{C e^{-c\Dres \bar{R}/2}, C e^{-c\Dres |Q|}\big\}$. 
Utilizing the elementary inequality $\min(x,y) \le \sqrt{xy}$, we obtain an upper bound $C e^{-\frac{c}{4}\Dres \bar{R} - \frac{c}{2}\Dres |Q|} \le C e^{-c'\Dres(\bar{R}+|Q|)}$ for $c' = c/4$. 
Conversely, for $\bar{R} < 2\sqrt{2}|Q|$, the individual bound \eqref{indiv_bound} alone implies the same decay rate, since $|Q| \ge (\bar{R} + |Q|)/4$. Consequently, there exists a modified absolute constant $c' > 0$ such that
\begin{align}
\label{diff_Fk_R}
&[\resol_z\big(\hH[u]^{\Omega_{\bar{R}}}(\xi)\big) \partial_{\hat u_{Q}} \hH[u]^{\Omega_{\bar{R}}}  \resol_z\big(\hH[u]^{\Omega_{\bar{R}}}(\xi)\big) ]_{\bzero,\bzero} \nonumber\\
&\hspace{5em}
-[\resol_z\big(\hH[u](\xi)\big) \partial{\hat u_{Q}} \hH[u]  \resol_z\big(\hH[u](\xi)\big) ]_{\bzero,\bzero} 
\le C e^{-c' \Dres ( \bar{R}+|Q|)}.
\end{align}
Combining this error \eqref{diff_Fk_R} with the estimate \eqref{proof-decay-dgH-xi} and $\bar{R}=\frac{1}{2}|\xi|$ yields 
\begin{align*}
&\left\vert \Fk(\xi)_{\bzero,\bzero} \right\vert \le C e^{-c\zeta |\xi|}  e^{-c \Dres |Q|} + C e^{-c' \Dres \bar{R}}e^{-c' \Dres |Q|} \le C e^{- c\Dres|\xi| - c \Dres |Q| } ,
\end{align*}
where $c'$ is absorbed into a generic constant $c$.
This completes the proof of \eqref{Fk:converge_G_xi}. 
\end{proof}

In the following lemma, we prove that the reciprocal-space integral of the shifted $\Fk$ is equivalent to its average trace, which is precisely evaluated as the limit of the traces of the discretized operators $\Fk^R$ with a modified square truncation $\bar \Omega_R := \{ (G_1,G_2)\in \RL^*_1 \times \RL^*_1: |G_1+G_2|\le R,  |G_1|\le R \}$,
\begin{align*}
    \Fk^{R}(\xi) = \frac{1}{2\pi \im} \oint_{\mathscr{C}_R} g(z) \resol_z\big(\hH[u]^{\bar\Omega_R}(\xi)\big)  \partial_{\hat u_{Q}} \hH[u]^{\bar\Omega_R}  \resol_z\big(\hH[u]^{\bar\Omega_R}(\xi)\big) \dd z.
\end{align*}
Furthermore, because $\bar \Omega_R$ scales equivalently to a standard circular truncation with respect to $R$, the decay estimate established in Lemma \ref{lemma:decay:Fk} remains valid for this modified domain via the same argument. In the following, we write $\Fk^R := \Fk^R(0)$ for simplicity.

\begin{lemma}[Trace averaging formula]
\label{lemma:integral of AB}
Under the assumptions of Theorem \ref{theorem:force}, then 
\begin{align}
\label{eq:integral of AB}
\int_{\R^{d}} \Fk(\xi)_{\bzero,\bzero} \dd \xi 
= \lim_{R \rightarrow \infty } \frac{|\Gamma_1^*||\Gamma_2^*|}{S_{d,R}}\Tr \Big(\Fk^{R} \Big)
\end{align}
where $S_{d,R}$ denotes the volume of a d-dimensional ball with radius R.
\end{lemma}

\begin{proof}

We define the following region 
\begin{align}
\label{def:SR}
S_R := \cup_{G_2 \in \RL^*_2 \cap B_R} \big(G_2 + \Gamma_2^*\big),
\end{align}
then we could obtain 
\begin{align}
\label{proof-commutation-a}
     &\left\vert \frac{|\Gamma_1^*||\Gamma_2^*|}{S_{d,R}} \Tr\Big( \Fk^{R} \Big) - \int_{\R^d} \Fk(\xi)_{\bzero,\bzero} \dd \xi \right\vert 
    ~\leq~ \frac{|\Gamma_1^*||\Gamma_2^*|}{S_{d,R}}\left\vert\sum_{\bG \in \bar\Omega_R} \Fk^{R}_{\bG,\bG} -\sum_{\bG\in \bar\Omega_R} \Fk_{\bG,\bG} \right\vert 
   \nonumber \\[1ex]
    +& ~\left\vert \frac{|\Gamma_1^*||\Gamma_2^*|}{S_{d,R}} \sum_{\bG\in \bar\Omega_R} \Fk_{\bG,\bG}  - \int_{S_R}  \Fk(\xi)_{\bzero,\bzero} \dd \xi  \right\vert
    ~+~ \left\vert \int_{S_R}  \Fk(\xi)_{\bzero,\bzero} \dd \xi  - \int_{\R^d}  \Fk(\xi)_{\bzero,\bzero} \dd \xi \right\vert
    \nonumber \\[1ex] 
    =: &~  T_1+T_2+T_3.
 \end{align}

To estimate the first term $T_1$, we have from the estimate \eqref{Fk_R} that for any $\bG=(G_1,G_2) \in \bar\Omega_R$,
\begin{align}
\label{proof-commutation-T1}
T_1 \leq \frac{|\Gamma_1^*||\Gamma_2^*|}{S_{d,R}} \Bigg( \sum_{\bG\in \bar\Omega_R} \left| \Fk^R_{\bG,\bG} - \Fk_{\bG,\bG} \right| \Bigg)
\leq  \frac{C}{S_{d,R}}  \sum_{\bG\in \bar\Omega_R}  e^{- c\Dres \big|R-|\bG|-\sqrt{2}|Q|\big| }
\leq C  R^{-d}.
\end{align}

To estimate the second term $T_2$ in \eqref{proof-commutation-a}, we first obtain  the ``symmetry" of Hamiltonian $\ham_{\bG,\bG} = \ham(G_1+G_2)_{\bzero,\bzero}$ from the definition \eqref{Hhat:xi}, and further have $\Fk_{\bG,\bG} = \Fk(G_1+G_2)_{\bzero,\bzero}$.
Then combine the definition of $S_R$ in \eqref{def:SR} that
\begin{align*}
T_2 & = \bigg| \frac{|\Gamma_1^*||\Gamma_2^*|}{S_{d,R}} \sum_{\bG\in \bar\Omega_R} \Fk(G_1+G_2)_{\bzero,\bzero} - \sum_{G_2\in\RL_2^*\cap B_R} \int_{\Gamma_2^*} \Fk(b+G_2)_{\bzero,\bzero}\dd b\bigg|
\\[1ex] 
& = \bigg| \frac{|\Gamma_1^*||\Gamma_2^*|}{S_{d,R}} \sum_{\ell\in\RL_1^*\cap B_{R}} \sum_{\substack{G_2\in\RL_2^* \\ |\ell+G_2|\le R}}  \Fk(\ell+G_2)_{\bzero,\bzero} -  \int_{\Gamma_2^*} \sum_{G_2\in\RL_2^*\cap B_R} \Fk(b+G_2)_{\bzero,\bzero}\dd b\bigg|  ,
\end{align*}
where the the definition \eqref{D_WL} has been used to get the second equality.
Let
\begin{eqnarray*}
f(\ell):= \sum_{G_2\in\RL_2^*,~ |\ell+G_2|\le R} \Fk(\ell +G_2)_{\bzero,\bzero}.
\end{eqnarray*}
We see that $f(\ell)$ is continuous and periodic with respect to $\RL_2^*$.
Note that $S_{d,R}$ is the volume of the $d$-dimensional ball with radius $R$, we have the fact $\frac{S_{d,R}}{|\Gamma_1^*|}=\#\big(\RL_1^* \cap B_{R} + \mathcal{O}(R^{d-1}) \big)$ where $\#(\cdot)$ denotes its cardinality.
Then we can apply Lemma B.3 in \cite{wang2025convergence} and derive that
\begin{align}
\label{proof-commutation-T2}
T_2 = \bigg| \frac{|\Gamma_1^*||\Gamma_2^*|}{S_{d,R}} \sum_{\ell\in\RL_1^* \cap B_{R}} f(\ell)
- \int_{\Gamma_2^*} f(b) \dd b \bigg| ,
\end{align}
which converges to 0 as $R \rightarrow \infty$.

For the last term $T_3$, we have from the estimate \eqref{Fk:converge_G_xi} that
\begin{equation}
\label{proof-commutation-T3}
T_3 \leq C \int_{\R^d \backslash S_R} e^{- c |\xi|-c \Dres|Q|} \dd \xi \leq C e^{- c R}.
\end{equation}

Finally, by combing \eqref{proof-commutation-a}, \eqref{proof-commutation-T1}, \eqref{proof-commutation-T2} and \eqref{proof-commutation-T3}, we can obtain
\begin{align}
    \lim_{R \rightarrow \infty } \left\vert \frac{|\Gamma_1^*||\Gamma_2^*|}{S_{d,R}}\Tr \Big(\Fk^{R} \Big)- \int_{\R^d} \Fk(\xi)_{\bzero,\bzero} \dd \xi\right\vert=0,
\end{align}
which complete the proof of \eqref{eq:integral of AB}.
\end{proof}

\begin{proof}
[Proof of Theorem \ref{theorem:force}]
Since the operators $\resol_z\big(\hH[u]^{\bar\Omega_R}(\xi)\big)$ and $\partial_{\hat u_{Q}} \hH[u]^{\bar\Omega_{R}}$ are finite dimensional matrices, the cyclic property of the trace yields
\begin{align}
\label{cyclicity}
\nonumber
    &\Tr \Big(\oint_{\mathscr{C}} g(z) \resol_z\big(\hH[u]^{\bar\Omega_R}(\xi)\big)  \partial_{\hat u_{Q}} \hH[u]^{ \bar\Omega_R}  \resol_z\big(\hH[u]^{\bar\Omega_R}(\xi)\big) \dd z \Big)\\
    =& \Tr \Big(\oint_{\mathscr{C}} g(z) \resol_z\big(\hH[u]^{\bar\Omega_R}(\xi)\big)^{2}  \partial_{\hat u_{Q}} \hH[u]^{ \bar\Omega_R} \dd z \Big).
\end{align}
By the same argument as in Lemma \ref{lemma:integral of AB}, we obtain
\begin{align}
\label{eq:integral of BA}
 &\int_{\R^{d}}  \oint_{\mathscr{C}} g(z) \Big[\resol_z\big(\hH[u](\xi)\big)^{2} \partial_{\hat u_{Q}} \hH[u] \Big]_{\bzero ,\bzero} \dd z \dd \xi \nonumber\\[1ex]
=  &\lim_{R \rightarrow \infty } \frac{|\Gamma_1^*||\Gamma_2^*|}{S_{d,R}}\Tr\Big(  \oint_{\mathscr{C}} g(z) \resol_z\big(\hH[u]^{\bar\Omega_R}(\xi)\big)^2 \partial_{\hat u_{Q}}  \hH[u]^{ \bar\Omega_R} \dd z \Big).
\end{align}
Taking the thermodynamic limit $R \rightarrow \infty$ and applying the discrete trace identity \eqref{cyclicity} alongside the limiting relations \eqref{eq:integral of BA} and \eqref{eq:integral of AB} allows us to transition back to the $(\bzero,\bzero)$-components of the continuous operators. This establishes the final equality
\begin{align*}
    &\int_{\R^d} \oint_{\mathscr{C}} g(z)\Big[ \resol_z\big(\hH[u](\xi)\big)  \partial_{\hat u_{Q}} \hH[u] \resol_z\big(\hH[u](\xi)\big) \Big]_{\bzero ,\bzero}\dd z \dd \xi\\
    =& \int_{\R^d} \oint_{\mathscr{C}} g(z)\Big[(z-\hH[u](\xi))^{2}  \partial_{\hat u_{Q}} \hH[u] \Big]_{\bzero ,\bzero} \dd z \dd \xi,
\end{align*}
which completes the proof.
\end{proof}

\subsection{Proof of Theorem \ref{theorem:u:convergence}}
\label{sec:proof:convergence}

To establish the convergence of our numerical scheme, we rely on a H\"{o}lder continuous formulation of the inverse function theorem, adapting the Lipschitz framework from \cite{ortner2011quasinonlocal}. Proved via a standard contraction mapping argument, this lemma will be directly applied to the truncated force operator $\nabla^{\ucut}_u E^{W,L}[u]$.
\begin{lemma}[Inverse Function Theorem with H\"{o}lder continuity]
\label{lem:ift}
Let $X, Y$ be Banach spaces, $A$ an open subset of $X$, and let $\F: A \to Y$ be Fr\'{e}chet differentiable. Suppose that $x_0 \in A$ satisfies
\begin{enumerate}[(1)]
    \item $\|\F(x_0)\|_Y \le \eta$,
    \item $\|D\F(x_0)^{-1}\|_{\mathcal{L}(Y,X)} \le \sigma$,
    \item $\|D\F(x_1) - D\F(x_2)\|_{\mathcal{L}(X,Y)} \le C \|x_1 - x_2\|_X^{1/2} \quad \text{for all} \quad x_1,x_2\in \overline B_X(x_0,2\eta\sigma)$,
    \item  $\overline{B}_X(x_0, 2\eta\sigma) \subset A$  and  $C \sigma^{3/2} \eta^{1/2} \le \frac{1}{2}$,
\end{enumerate}
then there exists a unique $x \in \overline{B}_X(x_0, 2\eta\sigma)$ such that $\F(x) = 0$ and $\|x - x_0\|_X \le 2\eta\sigma$.
\end{lemma}
\begin{proof}
We formulate the problem as finding a fixed point of the mapping $T(x) = x - D\mathcal{F}(x_0)^{-1}\mathcal{F}(x)$. A fixed point $x = T(x)$ exists if and only if $\mathcal{F}(x) = 0$. For any $x \in \overline{B}_X(x_0, r)$ with $r = 2\eta\sigma$, applying the $1/2$-H\"{o}lder continuity yields
\begin{align*}
    \|\mathcal{F}(x) - \mathcal{F}(x_0) - D\mathcal{F}(x_0)(x - x_0)\|_Y 
    & \le \int_0^1 C \|t(x-x_0)\|_X^{1/2} \|x-x_0\|_X \, dt 
    \\
    & \le C \|x-x_0\|_X^{3/2}.
\end{align*}
Consequently, the distance from $T(x)$ to the initial point $x_0$ is bounded by
\begin{align*}
\|T(x) - x_0\|_X \le \sigma\eta + \frac{2}{3}\sigma C \|x-x_0\|_X^{3/2} \le \sigma\eta + \frac{2}{3}\sigma C r^{3/2} = \sigma\eta \left( 1 + \frac{4\sqrt{2}}{3} C\sigma^{3/2}\eta^{1/2} \right).
\end{align*}
Utilizing the convergence condition $C\sigma^{3/2}\eta^{1/2} \le 1/2$, we obtain $1 + \frac{4\sqrt{2}}{3} C\sigma^{3/2}\eta^{1/2} \le 1 + \frac{2\sqrt{2}}{3} < 2$. This strict inequality ensures that $\|T(x) - x_0\|_X < 2\sigma\eta = r$, and therefore
\begin{align*}
    T(\overline{B}_X(x_0,r))\subset B_X(x_0,r)\subset \overline{B}_X(x_0,r).
\end{align*}
Next, for any $x, y \in \overline{B}_X(x_0, r)$, the convexity of the ball implies the line segment $y + t(x-y)$ remains within the ball for $t \in [0,1]$. We evaluate the difference:
\begin{align*}
    T(x) - T(y) = D\mathcal{F}(x_0)^{-1} \int_0^1 [D\mathcal{F}(x_0) - D\mathcal{F}(y+t(x-y))](x-y)  \dd t.
\end{align*}
Taking the norm and bounding the integrand by the maximum radius $C r^{1/2}$ gives
\begin{align*}
    \|T(x) - T(y)\|_X \le \sigma C r^{1/2} \|x-y\|_X = \sqrt{2} C \sigma^{3/2}\eta^{1/2} \|x-y\|_X.
\end{align*}
Given $C\sigma^{3/2}\eta^{1/2} \le 1/2$, the Lipschitz constant of $T$ is bounded by $\frac{\sqrt{2}}{2} < 1$. Thus, $T$ is a strict contraction mapping. 
By the Banach fixed-point theorem, there exists a unique fixed point in $\overline{B}_X(x_0,r)$ for $T$, and equivalently $\F(x)=0$ has a unique solution in this ball.
\end{proof}

To satisfy the residual control condition required by the inverse function theorem, it is necessary to bound the truncation errors associated with both the gradient and the Hessian. Since the convergence properties of these derivatives are fundamentally analogous to those of the energy $E$, we focus on the decay error of the energy for analytical simplicity. 
As an analogous fashion to Theorem \ref{theo:TDL}, the following lemma provides a sharper estimate for the convergence $E^{W,L}$ with respect to $W,~L$.
To quantify this decay, we introduce a more specific truncation of the wave vectors around $G$:
\begin{align*}
\DD_{W,L}(\bG)&:= \Big\{ \big(G_1',G_2'\big) \in\RL_1^* \times \RL_2^* :~ 
\big|(G_1'-G_1)+(G_2'-G_2)\big|\leq W, \nonumber\\
&\hspace{13em} \big|(G_1'-G_1)-(G_2'-G_2)\big|\leq L \Big\} .
\end{align*}

\begin{lemma}
\label{lemma:LDos}
Let $u\in\X^{\gamma}$and $g \in \setg$.
Then for $\bG,\bG' \in\RL_1^*\times\RL_2^*$, there exist positive constants $C,~c$ independent of $W, L$, such that
\begin{equation*}
   \Big| g\Big(\hH[u]^{\DD_{W,L}(\bG)}(\xi)\Big)_{\bG,\bG} - g\Big(\hH[u](\xi)\Big)_{\bG,\bG} \Big|
   \leq C \Big( e^{-c \zeta W} + e^{-c \Dres L} \Big) ,
\end{equation*}
where $\Dres = \min\{\tilde\gamma, \delta\}$ with $\tilde\gamma$ in Lemma \ref{lemma:decay of W}. Moreover, 
\begin{equation}
    \label{decay:E}
   \Big| E^{W,L}[u]-  E[u] \Big|
   \leq C \Big( e^{-c \zeta W} + e^{-c \Dres L} \Big).
\end{equation}
\end{lemma}

\begin{proof}
\label{sec:proof:lemma:shiftH}
Let $\Omega \subset \RL^*_1\times\RL^*_2$ such that $\Omega \supsetneq \DD_{W,L}$. We then expand the matrix $\hH[u]({\xi})^{\DD_{W,L}(\bG)}$ to a ``bigger" matrix by filling zero matrix elements,
\begin{eqnarray*}
\big(\hH[u]^{\DD_{W,L}(\bG)}_\Omega(\xi)\big)_{\bG',\bG''}=
\left\{
\begin{array}{ll}
\big(\hH[u]^{\DD_{W,L}(\bG)}(\xi)\big)_{\bG',\bG''} \quad &{\rm if \ \bG',\bG''\in \DD_{W,L}(\bG)},
\\[1ex]
\quad 0 & {\rm otherwise}.
\end{array}
\right.
\end{eqnarray*}

Let $\mathfrak{s}(g)$ represent the non-analytic region of the given function $g\in\setg$.
We can find a contour $\Cc_{W,L}\subset\C$ such that it encloses all the eigenvalues of $\hH[u]^\Omega(\xi)$ and $\hH[u]^{\DD_{W,L}}_\Omega(\bG)(\xi)$,
and satisfies
\begin{align*}
\min\Big\{
{\rm dist}\big(z,\mathfrak{s}(g)\big), ~
{\rm dist}\Big(z,\df\big(\hH[u]^{\DD_{W,L}(\bG)}_\Omega(\xi)\big)\Big), ~
{\rm dist}\Big(z,\df\big(\hH[u]^\Omega(\xi)\big)\Big),~
{\rm dist}\big(z, \{0\}\big) \Big\} \geq \frac{\delta}{2}    
\end{align*}
for any $z\in\Cc_{W,L}$, where $\df(\cdot)$ represents the spectrum set of an operator.

Using the decay of $\Big| \hat{W}_{j}[u_j](\bG) \Big|$ in Lemma \ref{lemma:decay of W}, we have for $\bG \neq \bG'$,
\begin{equation*}
    |\hH[u](\xi)_{\bG,\bG'}| \le C e^{- \frac{\alpha}{2} |\bG-\bG'|_+ -\tilde\gamma |\bG-\bG'|_-}.
\end{equation*}
Then, by using a Combes–Thomas type estimate like \eqref{decay:resol}, we have that there exist constants $c >0$ and $\Dres = \min\{\tilde\gamma, \delta\}$ such that for any $z \in \Cc_{W,L}$,
\begin{equation*}
    \big[\resol_z\big(\hH[u]^\Omega(\xi)\big) \big]_{\bG,\bG'}  \le C \frac{1}{\delta} e^{-c \delta|\bG-\bG'|_+ - c\Dres |\bG-\bG'|_-}.
\end{equation*}
This implies that,  for $\bG' \in \DD_{W/2,L/2}(\bG)$,
\begin{align*}
    &\Big|\big[\resol_z\big(\hH[u]^\Omega(\xi)\big)\big]_{\bG,\bG'} - \big[\resol_z\big(\hH[u]^{\DD_{W,L}(\bG)}_\Omega(\xi)\big) \big]_{\bG,\bG} \Big| 
    \\[1ex] \nonumber 
    \le & \sum_{\bG'',\bG''' \in \Omega }\Big| \big[ \resol_z \big(\hH[u]^{\DD_{W,L}(\bG)}_\Omega(\xi)\big) \big]_{\bG,\bG''} \big[\hH[u]^\Omega - \hH[u]^{\DD_{W,L}(\bG)}_\Omega \big]_{\bG'',\bG'''}  \big[ \resol_z\big(\hH[u]^\Omega(\xi)\big) \big]_{\bG''',\bG'} \Big|
    \\[1ex] \nonumber
    \leq & C\delta^{-2} \left( \sum_{\bG''\in \Omega \setminus \DD_{W,L}(\bG)}\sum_{\bG'''} + \sum_{\bG''\in \DD_{W,L}(\bG)}\sum_{\bG'''\in\Omega \setminus \DD_{W,L}(\bG)} \right) e^{-c\delta |\bG-\bG''|_+ -\Dres|\bG-\bG''|_-}
     \\[1ex] \nonumber
    & \hspace{12em}   \cdot e^{-\frac{\alpha}{2} |\bG''-\bG'''|_+ -\tilde\gamma |\bG''-\bG'''|_-} \cdot e^{-c\delta |\bG'''-\bG'|_+ -c\Dres|\bG'''-\bG'|_-}
    \\[1ex] \nonumber
    \leq &  C\delta^{-2}\left( \sum_{\bG''\in\Omega \setminus \DD_{W,L}(\bG)}\sum_{\bG'''} + \sum_{\bG''\in \DD_{W,L}(\bG)}\sum_{\bG'''\in\Omega \setminus \DD_{W,L}(\bG)} \right) 
    \left( e^{-c \delta |\bG-\bG''|_+ }
      \cdot e^{-\frac{\alpha}{2} |\bG''-\bG'''|_+}
       \right.
    \\[1ex] \nonumber
    &\hspace{8em}\left.
      \cdot e^{-c \delta|\bG'''-\bG'|_+} + e^{ -c\Dres |\bG-\bG''|_-}
        \cdot e^{ -\tilde\gamma |\bG''-\bG'''|_-}
        \cdot e^{ - c\Dres |\bG'''-\bG'|_-}
    \right)
    \\[1ex] \nonumber
    \leq & C \left(e^{-c \delta W}+e^{-c \Dres L} \right). 
\end{align*}
This combines with the contour integral representation yields
\begin{align*}
& \left|\lodoss{\hH[u]^\Omega(\xi)}{\bG}{\bG'} -\lodoss{\hH[u]^{\DD_{W,L}(\bG)}_\Omega(\xi)}{\bG}{\bG'} \right| 
\\[1ex]
= & \left| \frac{1}{2\pi \im}\oint_{\Cc_{W,L}} g(z) \bigg( \big(z-\hH[u]^\Omega(\xi)\big)^{-1} - \big(z-\hH[u]^{\DD_{W,L}(\bG)}_{\Omega}(\xi)\big)^{-1} \bigg)_{\bG,\bG'} \dd z \right|
\\[1ex]
\leq & C  \left(e^{-c \delta W}+e^{-c\Dres  L} \right) \oint_{\Cc_{W,L}} \big|g(z)\big| \dd z .
\end{align*}
Using the exponential decay of $g$ (indicated in the definition \eqref{deftestfunctiong}), we have that, as $W,~L\rightarrow\infty$, $\oint_{\Cc_{W,L}} \big|g(z)\big| \dd z$ is uniformly bounded by some constant depending on $\zeta$. Note that the above estimate holds for arbitrary $\Omega \supset \DD_{W,L}(\bG)$.
Therefore the limit $\lodoss{\hH^{\DD_{W,L}(\bG)}(\xi)}{\bG}{\bG'}$ exists, and it holds that
\begin{eqnarray}
\label{decay-WL-delta}
\left| g\Big(\hH[u]^{\DD_{W,L}(\bG)}(\xi) \Big)_{\bG,\bG'} - g\Big(\hH[u](\xi)\Big)_{\bG,\bG'} \right| \le C \big( e^{-c\delta W} + e^{-c\Dres L} \big) .
\end{eqnarray}
Furthermore, by invoking Lemma~B.2 of~\cite{wang2025convergence}
with $g \in \setg$, the convergence rate with respect to $W$ in \eqref{decay-WL-delta} can be further refined as
\begin{eqnarray}
\label{decay-WL-delta-zeta}
\left| g\Big(\hH[u]^{\DD_{W,L}(\bG)}(\xi) \Big)_{\bG,\bG} - g\Big(\hH[u](\xi)\Big)_{\bG,\bG} \right| \le C \big( e^{-c\zeta W} + e^{-c\Dres L} \big).
\end{eqnarray}
Combining the estimate \eqref{decay-xi-G1-G2} with \eqref{decay-WL-delta-zeta} derives
\begin{align}
\label{decay:xi-WL}
    \left| g\Big(\hH[u]^{\DD_{W,L}(\bG)}(\xi) \Big)_{\bG,\bG} - g\Big(\hH[u](\xi)\Big)_{\bG,\bG} \right| \le C e^{-c \zeta |\xi|} \big( e^{-c\zeta W} + e^{-c\Dres L} \big).
\end{align}
We then have from  \eqref{energy}, \eqref{energy:discrete} and the estimate \eqref{decay:xi-WL} with $\bG=\bzero$ that
\begin{align*}
    \Big| E^{W,L}[u]-  E[u] \Big|
    \le C \int_{\R^d} e^{-c \zeta |\xi|} \dd \xi \Big( e^{-c\zeta W} + e^{-c \Dres L} \Big)
   \le C  \Big( e^{-c\zeta W} + e^{-c \Dres L} \Big),
\end{align*}
which completes the proof. 
\end{proof}

The estimates for the first and second variations are derived by differentiating the same contour-integral representation with respect to the displacement variables. Comparing the integrand of $\partial_{\hat u_Q} E[u]$ with that of $E[u]$, the difference arises from the term $\partial_{\hat u_Q} \hH[u]$ defined in \eqref{form:dH}. This term inherits the analytic regularity of both $\hat{v}$ and $\hat{b}$, and the frequency shifts $Q$ do not alter the associated asymptotic decay rate in Fourier space. By employing an argument analogous to those used for \eqref{decay:Fk-Gi-Gj} and \eqref{Fk_R} in Lemma \ref{lemma:decay:Fk}, along with \eqref{decay:E} in Lemma \ref{lemma:LDos}, it immediately follows that 
\begin{align}
\label{decay_WL:F}
    \big|\partial_{\hat u_{Q}} E^{W,L}[u] - \partial_{\hat u_{Q}} E[u] \big| \le C \big( e^{-c\zeta W} + e^{-c \Dres \big|L-2|Q|\big|} \big).
\end{align}

For $Q, Q \in \RL^*_{\Fj}$,The second derivatives with respect to $\hat u_{Q}$ and $\hat u_{Q'}$ are given by
\begin{equation*}
    \partial^2_{\hat u_{Q},\hat u_{Q'}} \hH[u]_{\bG+\bG',\bG'} = \partial^2_{\hat u_{Q},\hat u_{Q'}} \hat{W}_{j}[u_j](\bG) = -(w)^2 \frac{\hat v(w)}{|\Gamma_j|} \cdot \hat b[w, u_j](G_{\Fj} - Q - Q').
\end{equation*}
where $w=G_1+G_2$. An analogous derivation yields
\begin{align}
\label{decay_WL:DF}
    |\partial^2_{\hat u_{Q},\hat u_{Q'}} E^{W,L}[u] - \partial^2_{\hat u_{Q},\hat u_{Q'}} E[u]| \le C \Big( e^{-c\zeta W} + e^{-c \Dres \big| L - 2|Q+Q'| \big|} \Big).
\end{align}

The estimate of the residual of the numerical scheme also requires the decay of the gradient of $E[u]$ in frequency space. The following lemma provides the exponential decay bound.
\begin{lemma}
\label{lemma:F_decay_u}
Let $u\in\X^{\gamma}$ and $g\in \setg$. There exist positive constants $C$ dependent on $\zeta,\delta, \gamma$ and $c$ independent of $\gamma,\delta$ such that
\begin{align}
\label{decay:Fk}
     \left|\partial_{\hat u_{Q}} E[u] \right| \le C e^{-c \Dres |Q| }.
\end{align}
\end{lemma}
\begin{proof}
By applying the estimate \eqref{Fk:converge_G_xi} from Lemma \ref{lemma:decay:Fk}, we obtain:
\begin{align*}
    \left|\partial_{\hat u_{Q}} E[u] \right| \le \int_{\R^d} C e^{-c|\xi|-c \Dres |Q|} \dd \xi \le C e^{-c \Dres |Q|} \int_{\R^d} e^{-c\Dres|\xi|} \dd \xi.
\end{align*}
Since the term $e^{-c\Dres|\xi|}$ decays exponentially with respect to $\xi$, its integral over $\R^d$ is bounded and can be absorbed into the generic constant $C$. This directly yields the desired estimate \eqref{decay:Fk}, completing the proof.
\end{proof}
Since the truncated operator $\partial_{\hat u_{Q}} E^{W,L}[u]$ naturally preserves the exponential decay established in Lemma \ref{lemma:F_decay_u}, their difference is also bounded by $\mathcal{O}(e^{-c \Dres |Q|})$. By performing a standard exponential interpolation between this intrinsic decay bound and the truncation error estimate in \eqref{decay_WL:F}, and utilizing the basic inequality $\big|L-2|Q|\big| \ge L - 2|Q|$ to absorb the extra $|Q|$ growth into the generic constant $c$, we deduce the refined difference bound:
\begin{align}
\label{decay_WL_Q:F}
    \big|\partial_{\hat u_{Q}} E^{W,L}[u] - \partial_{\hat u_{Q}} E[u] \big| \le C e^{-c \Dres |Q| } \big( e^{-c\zeta W} + e^{-c \Dres L} \big).
\end{align}

Beyond the residual analysis, establishing coercivity relies on the decay estimates for the second variations with respect to displacements. The next lemma establishes a rigorous decay bound for the Hessian, thereby confirming its boundedness as an operator.
\begin{lemma}
\label{lemma:DF_decay_U}
Let $u\in\X^{\gamma}$ and $g \in \setg$. Under the assumptions of Theorem \ref{theorem:u:convergence}, there exist positive constants $C$ dependent on $\zeta,\delta,\gamma$ and $c$ independent of $\gamma,\delta$ such that
\begin{align}
\label{decay_u:DF}
\left|\partial^2_{\hat u_{Q},\hat u_{Q'}} E[u]
\right|
\le C  e^{-c\Dres |Q+Q'|} .
\end{align}
Thus, $ \nabla^2_{u} E[u] := \{\partial^2_{\hat u_{Q},\hat u_{Q'}} E[u] \}_{Q,Q' \in (\RL^*_{1} \cup \RL^*_{2})\setminus \{0\} } $ is a bounded operator on $\ell^2$ with its operator norm given by
\begin{align}
    \Vert \nabla^2_{u} E[u] \Vert= \sup_{\nu \in \ell^2 \setminus \{0\}} \frac{|\nabla^2_{u}E[u]\cdot \nu |}{|v|}. 
\end{align}
\end{lemma}
\begin{proof}
To evaluate the Hessian operator $\nabla^2 E[u]$, we utilize the contour integral representation of the energy functional. Without loss of generality, we may assume that $Q, Q'$ belong to the same $\RL^*_{j}$. By differentiating the gradient formula \eqref{form:F}, we obtain
\begin{align*}
     \partial^2_{\hat u_{Q},\hat u_{Q'}} E[u] 
     =& \partial_{\hat u_{Q'}} \int_{\R^d} \frac{1}{2\pi \im} \oint_{\mathscr{C}} g(z) \big[ \resol_z\big(\hH[u](\xi)\big) \partial_{\hat u_{Q}} \hH[u]\resol_z\big(\hH[u](\xi)\big)\big]_{\bzero,\bzero} \dd z \dd \xi\\
     =& \int_{\R^d} \frac{1}{2\pi \im} \oint_{\mathscr{C}} g(z) \bigg( \Big[  \resol_z\big(\hH[u](\xi)\big) \partial^2_{\hat u_{Q},\hat u_{Q'}} \hH[u]\resol_z\big(\hH[u](\xi)\big)\Big]_{\bzero,\bzero} \\
     &+\Big[ \resol_z\big(\hH[u](\xi)\big) \partial_{\hat u_{Q}}  \hH[u]\resol_z\big(\hH[u](\xi)\big) \partial_{\hat u_{Q'}} \hH[u]\resol_z\big(\hH[u](\xi)\big)\Big]_{\bzero,\bzero}\\
     &+\Big[ \resol_z\big(\hH[u](\xi)\big) \partial_{\hat u_{Q}}  \hH[u]\resol_z\big(\hH[u](\xi)\big) \partial_{\hat u_{Q'}} \hH[u]\resol_z\big(\hH[u](\xi)\big)\Big]_{\bzero,\bzero} \bigg)
     \dd z \dd \xi,\\
     =:& \int_{\R^d} \DFkq(\xi) \dd \xi.
\end{align*}
By an argument analogous to the proof of Lemma \ref{lemma:decay:Fk}, we obtain 
\begin{align*}
    |\DFkq(\xi)_{\bzero,\bzero}| \le C e^{- c\Dres|\xi| - c \Dres |Q+Q'| }.
\end{align*}
Consequently, it follows that
\begin{align*}
    \left|\partial_{\hat u_{Q}, \hat u_{Q'}} E[u] \right| \le \int_{\R^d} C e^{-c\delta_\gamma|\xi|-c \Dres |Q+Q'|} \dd \xi \le C e^{-c \Dres |Q+Q'|} .
\end{align*}
Here, the exponential decay of $e^{-c\Dres|\xi|}$ with respect to $\xi$ guarantees a bounded integral over $\R^d$, which is naturally absorbed into the generic constant $C$.

Specifically, the exponential decay in \eqref{decay_u:DF} guarantees 
\begin{align*}
    \max_{Q} \big(\sum_{Q'} C  |\partial^2_{\hat u_{Q},\hat u_{Q'}} E[u]|\big) \le \max_{Q} \big(\sum_{Q'} C e^{-c\Dres |Q+Q'|} \big) \le C,\\
    \max_{Q'} \big(\sum_{Q} C  |\partial^2_{\hat u_{Q},\hat u_{Q'}} E[u]|\big) \le \max_{Q'} \big(\sum_{Q} C e^{-c\Dres |Q+Q'|} \big) \le C,
\end{align*}
which together with Schur's test (see, e.g., \cite[Theorem 6.18]{folland1999real}) implies that
\begin{equation*}
|\nabla^2_{u}E[u]\cdot v | \le C | v | \qquad\forall~v \in \ell^2 .
\end{equation*}
This completes the proof.
\end{proof}

By a similar estimation combined with Lemma \ref{lemma:DF_decay_U} and the estimate \eqref{decay_WL:DF}, we obtain the corresponding difference bound for the second derivatives:
\begin{align}
\label{decay_WL_Q:DF}
    \big|\partial^2_{\hat u_{Q},\hat u_{Q'}} E^{W,L}[u] - \partial^2_{\hat u_{Q},\hat u_{Q'}} E[u] \big| \le C  e^{-c\Dres |Q+Q'|} \Big( e^{-c\zeta W} + e^{-c \Dres L} \Big).
\end{align}
The application of Lemma \ref{lem:ift} requires the Hessian to be H\"{o}lder continuous, which we formalize in the following lemma.
\begin{lemma}[H\"{o}lder continuity of the gradient and Hessian]
\label{lemma:continuity}
Let $u, u' \in \X^\gamma$ and $g \in \setg$. For any $\bG=(G_1,G_2)\in\RL_1^*\times\RL_2^*$ and $j=1,2$, there exists a constant $C>0$ independent of $u$ and $u'$, such that
\begin{equation}
\label{continue:F}
    |\nabla_u E[u]-\nabla_u E[u']| \le C|u-u'|^{\frac{1}{2}}
\end{equation}
and
\begin{equation}
\label{continue:DF}
    \|\nabla_u^2 E[u]-\nabla_u^2 E[u']\| \le C|u-u'|^{\frac{1}{2}},
\end{equation}
where $|\cdot|$ denotes the $L^2(\Gamma_j)$-norm.
\end{lemma}

\begin{proof}
We first estimate the difference $\Delta \hat{b} := \hat b[w, u](G_{j}) - \hat b[w, u'](G_{j})$. Using the elementary inequality $|e^{ix} - e^{iy}| \le |x-y|$ for $x,y \in \R$, we have
\begin{equation}
\label{bound:diff_L1}
    |\Delta \hat{b}| \le \frac{1}{|\Gamma_{j}|} \int_{\Gamma_{j}} \Big| e^{-i w\cdot u(s)} - e^{-i w\cdot u'(s)} \Big| \dd s \le \frac{|w|}{|\Gamma_{j}|} \int_{\Gamma_{j}} |u(s) - u'(s)| \dd s.
\end{equation}
Applying the Cauchy-Schwarz inequality over the bounded domain $\Gamma_j$ yields
\begin{equation}
\label{bound:diff_L2}
    |\Delta \hat{b}| \le |\Gamma_{j}|^{-\frac{1}{2}} |w| |u - u'|.
\end{equation}
Interpolating \eqref{bound:diff_L2} with the decay estimate \eqref{decay:b}, we obtain:
\begin{equation}
\label{bound:interpolated_b}
    |\Delta \hat{b}| \le C |w|^{\frac{1}{2}} e^{\frac{1}{2} M_u |w|} e^{-\frac{1}{2}\gamma |G_{j}|} |u - u'|^{\frac{1}{2}}.
\end{equation}
Substituting \eqref{bound:interpolated_b} into \eqref{form:W} gives the continuity of the auxiliary potential $\hat{W}_{j}$
\begin{equation}
\label{continuity:W}
   \Big| \hat{W}_{j}[u](\bG) - \hat{W}_{j}[u'](\bG) \Big| \le C e^{-\frac{\alpha}{2}|G_1 + G_2| - \frac{\gamma}{2} |G_1 - G_2|} |u - u'|^{\frac{1}{2}}.
\end{equation}
Since $\hat{H}[u]$ inherits the identical continuity and decay properties from $\hat{W}_{j}[u]$, the regularities can be directly extended to the energy functional. Specifically, combining \eqref{continuity:W} with the estimate \eqref{decay:Fk} in Lemma \ref{lemma:F_decay_u}, we establish \eqref{continue:F}:
\begin{equation*}
    |\nabla_u E[u]-\nabla_u E[u']| \le C \big(\sum_Q e^{-2c \Dres |Q| }\big)^\frac{1}{2} |u-u'|^{\frac{1}{2}} \le C |u-u'|^{\frac{1}{2}}.
\end{equation*}
Analogously, applying the estimate \eqref{decay_u:DF} and the operator norm bound from Lemma \ref{lemma:DF_decay_U} yields \eqref{continue:DF}.
\end{proof}

\begin{proof}[Proof of Theorem \ref{theorem:u:convergence}]
To establish the convergence of the scattering-channel approximation, we introduce the truncation operator
\begin{align*}
    \Pi_{\ucut} ~:~ \nu =\{\nu_{Q}\}_{Q \in (\RL^*_1 \cup \RL^*_2)\setminus\{0\}} \rightarrow \nu=
     \left\{  
\begin{array}{lr}  
\nu_{Q},  \quad |Q| \le \ucut
\\[1ex]  
0,  \quad otherwise.
\end{array}
\right. 
\end{align*}
which projects onto the frequency modes $|Q| \le \ucut$.
We apply the Inverse Function Theorem (Lemma \ref{lem:ift}) by setting our reference point as $x_0 := \Pi_{\ucut} \bar u$. The continuous force operator is defined as $\F(u) := \nabla_u E[u]$, while the discrete truncated force operator is given by $\F^{W,L,\ucut}(u) := \nabla_u^\ucut E^{W,L}[u]$. The exact numerical solution satisfies $\F^{W,L,\ucut}(\bar u^{W,L,\ucut}) = 0$.

To estimate the initial residual $\eta$ at $x_0$, we use $\F(\bar u) = 0$. Applying the triangle inequality in conjunction with \eqref{decay_WL_Q:F}, \eqref{decay:Fk} in Lemma \ref{lemma:F_decay_u}, and the continuity property \eqref{continue:F} yields the bound
\begin{align}
\label{eta:combined}
    | \F^{W,L,\ucut}(x_0) | 
    &\le | \F^{W,L,\ucut}(x_0) - \F^\ucut(x_0) | + | \F^\ucut(x_0) - \F(x_0) | + | \F(x_0) - \F(\bar u) | \nonumber\\[1ex]
    &\le C\Big(e^{-c \zeta W} + e^{- c \Dres L}\Big) + \sum_{|Q|>\ucut} C e^{-c\Dres |Q|} + C | \Pi^\ucut\bar u - \bar u |^\frac{1}{2} \nonumber\\[1ex]
    &\le C\Big(e^{-c \zeta W} + e^{- c \Dres L} + e^{-c \Dres \ucut}\Big) =: \eta.
\end{align}

Next, we establish the uniform invertibility of the Jacobian at $x_0$. By Assumption \eqref{regularity_stability}, the Hessian $\nabla^2_u E[\bar u]$ is coercive with a constant $\vartheta > 0$. For any $v \in \X$ with its Fourier sequence $\hat v \in \ell^2(\RL^*_1 \cup \RL^*_2)$, the quadratic form of $(\nabla^\ucut_u)^2 E^{W,L} [x_0]$ is decomposed as
\begin{align*}
    \big\langle (\nabla^\ucut_u)^2 E^{W,L} [x_0] \hat v, \hat v \big\rangle 
    &= \big\langle \nabla_u^2 E[\bar u] \hat v, \hat v \big\rangle + \big\langle \big( (\nabla^\ucut_u)^2 E^{W,L}[x_0] - (\nabla^\ucut_u)^2 E[x_0] \big)\hat v, \hat v \big\rangle \\[1ex]
    &+ \big\langle \big( (\nabla^\ucut_u)^2 E[x_0] - \nabla_u^2 E[x_0] \big)\hat v, \hat v \big\rangle + \big\langle \big( \nabla_u^2 E[x_0] - \nabla_u^2 E[\bar u] \big) \hat v, \hat v \big\rangle \\[1ex]
    & \ge  \vartheta |\hat v|^2 - | (\nabla^\ucut_u)^2 E^{W,L}[x_0] - (\nabla^\ucut_u)^2 E[x_0] | |\hat v|^2 \\[1ex]
    &- \left| \langle \big( (\nabla^\ucut_u)^2 E[x_0] - \nabla_u^2 E[x_0] \big)\hat v, \hat v \rangle \right| - \left| \langle \big( \nabla_u^2 E[x_0] - \nabla_u^2 E[\bar u] \big) \hat v, \hat v \rangle \right|.
\end{align*}
The three perturbation terms correspond to the scattering-channel truncation error, the Fourier truncation error, and the discrepancy between $x_0$ and $\bar u$, respectively. Relying on the Hessian element convergence \eqref{decay_WL_Q:DF}, the decay estimate \eqref{decay_u:DF}, and the continuity \eqref{continue:F}, we can strictly bound each perturbation term by $\frac{\vartheta}{6} |\hat v|^2$ for sufficiently large $W$, $L$, and $\ucut$. Consequently, the strict coercivity of the discrete operator is preserved:
\begin{equation*}
    \langle (\nabla^\ucut_u)^2 E^{W,L}[x_0] \hat v, \hat v \rangle \ge \left( \vartheta - \frac{\vartheta}{6} - \frac{\vartheta}{6} - \frac{\vartheta}{6} \right) |\hat v|^2 = \frac{\vartheta}{2} |\hat v|^2.
\end{equation*}
This implies the Jacobian is uniformly invertible with the bound
\begin{equation}
\label{sigma:combined}
    |(\nabla^\ucut_u)^2 E^{W,L}(x_0)^{-1}| \le \frac{1}{2\vartheta} =:\sigma.
\end{equation}

Finally, given the established bounds \eqref{eta:combined} and \eqref{sigma:combined}, we observe that the residual $\eta$ decays exponentially with respect to $W, L$, and $\ucut$, whereas the norm of the inverse operator $\sigma$ remains bounded by a constant independent of the truncation parameters. Therefore, for sufficiently large parameters, the condition $C \sigma^{3/2}\eta^{1/2} \le \frac{1}{2}$ is trivially satisfied. An application of Lemma \ref{lem:ift} guarantees the existence of a unique exact numerical solution $\bar{u}^{W,L,\ucut}$ in the neighborhood of $x_0$. 

The total convergence error is thus concluded by
\begin{align*}
    | \bar u - \bar u^{W,L,\ucut} | &\le | \bar u - x_0 | + | x_0 - \bar{u}^{W,L,\ucut} | \le C e^{-c \Dres \ucut} + 2 \eta \sigma \\
    &\le \frac{C}{\vartheta} \Big( e^{-c \zeta W} + e^{-c \Dres L }+e^{-c \Dres \ucut}\Big).
\end{align*}
This completes the proof.
\end{proof}

\small
\bibliographystyle{plain} 
\bibliography{bib}

\end{document}